\newif\ifSC
\ifSC
\documentclass[onecolumn,draftclsnofoot,12pt]{IEEEtran}
\else
\documentclass[10pt,twocolumn]{IEEEtran}
\fi
\usepackage{soul}
\usepackage{siunitx}
\usepackage{cite}
\usepackage{amsmath,amssymb,amsfonts}
\usepackage{algorithmic}
\usepackage{xcolor}
\usepackage{amssymb}
\usepackage{bbm}
\usepackage{bm}
\usepackage{amssymb,mathtools,amsmath}
\usepackage{amsbsy}
\usepackage{multicol}
\usepackage{fontenc}
\usepackage{comment} 
\usepackage{booktabs}

\usepackage[font=small]{subcaption}
\usepackage[font=small]{caption}
\usepackage{graphicx} 

\def\BibTeX{{\rm B\kern-.05em{\sc i\kern-.025em b}\kern-.08em
		T\kern-.1667em\lower.7ex\hbox{E}\kern-.125emX}}
\newcounter{relctr} 
\everydisplay\expandafter{\the\everydisplay\setcounter{relctr}{0}} 

\newcommand\numeq[1]%
{\stackrel{\scriptscriptstyle(\mkern-1.5mu#1\mkern-1.5mu)}{=}}
\AtBeginDocument{}

\newcommand{\z}{\bm{z}}

\newcommand{\road}{\bm{l}}

\newcommand{\1}{\mathbbm{1}}
\newcommand{\pt}{\mathrm{P}}
\newcommand{\prm}{\mathrm{p}}
\newcommand{\rsu}{\mathrm{r}}
\newcommand{\avg}{\mathrm{avg}}
\newcommand{\drm}{\mathrm{d}}

\newcommand{\ob}{\mathrm{o}}

\newcommand{\dv}{\mathrm{d}}
\newcommand{\bt}{\mathbf{b}}

\newcommand{\mrm}{\mathrm{m}}
\newcommand{\Nrm}{\mathrm{N}}
\newcommand{\irm}{\mathrm{i}}

\newcommand{\trm}{\mathrm{t}}

\newcommand{\brm}{\mathrm{b}}

\newcommand{\A}{\mathcal{A}}

\newcommand{\R}{\mathbb{R}}

\newcommand{\matern}{Mat\'ern~}

\definecolor{ForestGreen}{RGB}{34,139,34}
\newcommand{\SNR}{\mathrm{SNR}}

\usepackage{bbm}
\usepackage{amssymb}
\usepackage{amsmath,amsthm}
\usepackage{color}
\usepackage{amsfonts}
\usepackage{graphicx}
\usepackage{verbatim}
\usepackage{epstopdf}
\usepackage{cite}
\usepackage{tabulary}
\usepackage{mathtools}
\usepackage{enumitem}

\newcommand\expect[1]{\mathbb{E}\left[#1\right]}
\newcommand\prob[1]{\mathbb{P}\left[#1\right]}

\newcommand{\SINR}{\text{SINR}}

\newcommand{\expects}[2]{\mathbb{E}_{#1}\left[#2\right] }

\newcommand{\laplace}[1]{\mathcal{L}_{#1} }
\newcommand{\ie}{{\em i.e.}~}

\newcommand{\pu}{\mathrm{P}}

\newcommand\lambdaPT{\lambda_\pu}

\newcommand{\X}{\mathbf{X}}

\newcommand{\expS}[1]{\exp{\left(#1\right)}}

\def\home{\hbox{\kern3pt \vbox to13pt{}%
   \pdfliteral{q 0 0 m 0 5 l 5 10 l 10 5 l 10 0 l 7 0 l 7 5 l 3 5 l 3 0 l f
               1 j 1 J -2 5 m 5 12 l 12 5 l S Q }%
   \kern 13pt}}

\newtheorem{theorem}{Theorem}
\newtheorem{lemma}{Lemma}
\newtheorem{definition}{Definition}
\newtheorem{corollary}{Corollary}[theorem]

\graphicspath{{figs/}}
	
	\title{A novel framework for modeling and analysis of platoon based movement of VUs}
	\author{Kaushlendra Pandey, Harpreet S. Dhillon, Abhishek K. Gupta\,\,
	\vspace{-2.1em}
		\thanks{K. Pandey and A. K. Gupta are with IIT Kanpur, India, 208016. Email:\{kpandey,gkrabhi\}@iitk.ac.in. H. S. Dhillon is with Wireless@VT, Bradley Department of Electrical and Computer Engineering, Virginia Tech, Blacksburg, VA 24061. (Email: hdhillon@vt.edu). A. K. Gupta sincerely acknowledges the support received from the grants CRG/2023/5206 and Qualcomm 6GUR.} 
	}

\usepackage{graphicx,amsmath,amssymb}
\allowdisplaybreaks
\begin{document}

\title{On the Analysis of Platooned Vehicular Networks on Highways}
\maketitle
\begin{abstract}
Vehicular platooning refers to coordinated and close movement of vehicular users (VUs) traveling together along a common route segment, offering strategic benefits
such as reduced fuel costs, lower emissions, and improved traffic
flow. {Highways offer a natural setting for platooning due to extended travel distances, yet their potential remains underexplored, particularly in terms of communication and connectivity. Given that effective platooning relies on robust vehicle-to-vehicle (V2V) and vehicle-to-infrastructure (V2I) links, understanding connectivity dynamics on highways is essential.} In this paper, we analyze the dynamics of vehicular platooning on a highway and its impact on the performance of two forms of vehicular communications- namely V2V and V2I communication- compared to independent vehicle movement on a highway. The vehicular networks consists of road-side units (RSUs), modeled as a  1D Poisson point process (PPP), to provide the vehicular connectivity to the VUs. VUs are modeled as 1D
PPP under the non-platooned traffic scenario (N-PTS) and as a 1D \matern cluster process (MCP) under the platooned traffic scenario (PTS). We evaluate the distribution on the per-RSU load, representing the number of VUs served, for the typical and tagged RSU. Additionally, we derive coverage probability (CP) and rate coverage (RC), which measures the probability of the signal-to-interference-plus-noise ratio (SINR) and achievable rate above a specified threshold at the typical VU along with their meta distribution (MD), providing a deeper understanding of
the reliability and variability of these metrics across different
spatial distributions of VUs and RSUs. Finally, we validate our
theoretical findings through simulations and provide numerical
insights into the impact of different traffic patterns on RSU load distribution, CP, and RC performance.
\end{abstract}


\section{Introduction}
Vehicular platooning offers a cost-effective solution to key traffic challenges such as congestion, energy waste, and accidents \cite{jia2015survey,hussein2021vehicle}. In a PTS, VUs are expected to coordinate seamlessly to maintain formation and achieve a common goal. Apart from aiding in the management of increasing traffic congestion, vehicular platooning may enhance connectivity for VUs by facilitating V2V communication \cite{perfecto2017millimeter}, thus reducing the reliance on the cellular infrastructure network consisting of RSU. However, because vehicles in a platoon are positioned closely together, this arrangement may place additional load on RSUs. The term load here refers to the number of VUs served by the typical RSU \cite{pandey2023properties}. While under-loaded RSUs result in underutilization of resources, overloaded RSUs may be unable to meet the connectivity demands of all VUs \cite{pandey2023fundamentals}. The resulting load distribution directly influences per-VU throughput, a key factor for maintaining stable platoon operations. Therefore, understanding how vehicular platooning affects RSU load distribution is critical. This paper addresses this issue by analyzing the RSU load distribution under both platooned and non-platooned vehicle movements on a highway.\\
{\textit{Related work:}
With the growing importance of vehicular communications, recent studies have analyzed their performance under diverse traffic models and scenarios~\cite{huang2020geometry,sun2020distributed,kimura2021performance}. To evaluate network performance from a system-level perspective--particularly the effects of load distribution on SINR and achievable rates--stochastic geometry provides a tractable framework for capturing complex spatial dynamics~\cite{dhillon2020poisson,chetlur2020load,choi2018poisson}. For highway VU locations are modeled with a $1$D PPP~\cite{blaszczyszyn2013stochastic}, enabling derivation of key metrics such as the SINR distribution. Utilizing the similar models, the works in \cite{1dchapter} and \cite{cheng2020connectivity} analyzed the connectivity and capacity in vehicular networks. With increasing model complexity, recent work has considered alternative point processes (PPs) to model VU locations, for example, the \matern process, which introduces headway distances between VU pairs to study V2V safety communications \cite{tong2016stochastic}. Further, \cite{koufos2019meta} analyzed the MD of SINR, showing its usefulness in explaining disparities in link success probabilities. In \cite{ni2015interference} authors extended the analysis to multi lane scenarios by incorporating vehicle-following models to derive link and network capacities. \\Turning to networks under PTS, \cite{shao2015performance} analyzed connectivity in platooned vehicular traffic, highlighting the impact of traffic density. Further, \cite{wang2022design} presented a deep reinforcement learning-based cruise control system to mitigate communication delays and dynamic disturbances in vehicular adhoc networks. Platooning changes how load is distributed across RSUs and affects the rate distributions of VUs. These dynamics are analytically captured using a stochastic geometry framework in \cite{pandey2023fundamentals}. 
	However, the study of load distribution analysis for highways remains
	relatively underexplored.
	For example, \cite{pandey2023properties,pandey2023fundamentals} provided approximate analyses of load distribution in two-dimensional vehicular networks, where vehicle locations were modeled using the Poisson line Cox process \cite{choi2018poisson}. These models, while insightful, may lead to slightly complex load formulations, which may limit their analytical tractability. In vehicular platooning, vehicular communication encompasses diverse technologies, such as cellular V2I connectivity and V2V safety messaging, which may coexist within the same network. The current literature lacks a comprehensive analysis of various communications that may be present in such networks in terms of distribution of per-RSU load, SINR and rate along with  their MDs. MDs, in particular, hold great potential for quantifying link reliability under varying SINR thresholds. Consequently, understanding the impact of platooning on the RSU load distribution is essential. 
	
	This paper attempts to fill  above-mentioned gaps by analyzing the
	load distribution, coverage probability, RC, and their MDs
	for both platooning and non-platooned vehicle
	movements on a road or highway.

\textit{Contribution:} The key contributions of the paper are summarized next.

\begin{itemize}
	\item We derive closed-form expressions for the load distribution on RSUs under PTS and N-PTS, highlighting that the difference in RSU load arises from the correlation of VU locations in PTS. We further analyze the load distribution by deriving key metrics—mean, variance, and skewness—which offer insight into how load variation affects RSU performance.
	
	\item Using the load distribution, we derive the RSU active probability and the active RSU density under PTS and N-PTS. We present a framework illustrating various forms of communications within a vehicular network. We first examine  V2V communication and
	analyze the vehicular connectivity of the typical VU through a performance metric termed {\em connectivity degree}.
	\item We study V2I connectivity in a infrastructure network consisting of RSUs and derive the CP for the typical VU in two scenarios. To understand link reliability, we derive exact expressions for the MD of the SINR and show how VU density affects CP. We also derive the rate distribution and its MD under PTS and N-PTS to show how platooning and VU density affect the data rate. 
	\item Through numerical results, we demonstrate key performance trends with respect to VU density and platooning extent. Our analysis shows that PTS improves RSU efficiency due to increased burstiness in load compared to N-PTS. These results support our contribution by revealing that higher VU density increases the RSU efficiency under PTS.
	
\end{itemize}
{\textit{Notation:}
 The probability generating function (PGF) of a non-negative integer-valued random variable $S$ is denoted by $\mathcal{P}_S(\cdot)$. The origin is denoted by $\ob$, and the one-dimensional ball of radius $a$ centered at the origin is denoted by $\bt_{1}(\ob, a)$.
 	 The upper incomplete Gamma function and the Gamma function are denoted by $\Gamma_{\mathrm{u}}(a, x)$ and $\Gamma(x)$, respectively. Further $\beta(r)=2\min(r,a)$ and $\bar{\beta}(r)={\beta(r)}/{(2a)}$. The skewness of an RV $X$ is
	 \begin{align}\label{skew}
&\mathbb{S}[X]=\frac{\mathbb{E}\left[X^{3}\right]-3\mathbb{E}\left[X\right]\mathrm{Var}[X]-\left(\mathbb{E}[X]\right)^{3}}{(\mathrm{Var}[X])^{\frac{3}{2}}}.
\end{align} For a RV $X$ with PGF $\mathcal{P}_{X}(s)$, the  mean, variance and third moment (essential for deriving the skewness) are given as 
 		 \begin{align}
	&\mathbb{E}[X]=[\mathcal{P}^{(1)}_{X}(s)]_{s=1},\\
	&\mathrm{Var}[X]=[\mathcal{P}^{(2)}_{X}(s)]_{s=1}+\mathbb{E}\left[X\right]-(\mathbb{E}\left[X\right])^{2}\label{variance},\\
	&\mathbb{E}\left[X^{3}\right]=\left[\mathcal{P}^{(3)}_{X}(s)\right]_{s=1} \!\!+ 3\left(\mathrm{Var}[X]+\left(\mathbb{E}[X]\right)^{2}\right) -2\mathbb{E}[X]\nonumber\\
	&=\left[\mathcal{P}^{(3)}_{X}(s)\right]_{s=1}+3\left[\mathcal{P}^{(2)}_{X}(s)\right]_{s=1}+\mathbb{E}[X].\label{third-mean}
	\end{align}}
 A ramp $\mathcal{R}(x-a)$ and a unit step function $u(x)$ are defined as
 \begin{align}
	\mathcal{R}(x-a)&=(x-a)\1( x>a),\label{rampsignal}\\
	u(x)&=1 \1(x\geq0).\label{unitstepsignal}
 \end{align}The notation $\bar{\A}(r,a,x)={\A(r,a,x)}/{(2 a)},$ where $\A(r,a,x)$ denote the length of intersection of two one dimensional balls of radius $r$ and $a$ with centers $x$ distance apart, which is given as
\begin{align}
	&= 
	\begin{cases}
		2\min(r,a)\stackrel{\Delta}{=}\beta(r,a),& \text{if } 0<x<|r-a|,\\
		r+a-x,              & \text{if } |r-a|<x<r+a.
	\end{cases}
\end{align} Let $F(m,k)\!\!=\!\!\int_{0}^{2a}\! x^{k}e^{-mx}\dv x\text{ and }G(m,k)=\int_{2a}^{\infty}x^{k}e^{-mx}\dv x$, where,  \begin{align}
&F(m,k)={\left[\Gamma(k+1)-\Gamma_{\rm u}((k+1),2ma)\right]}/{m^{k+1}},\label{funF}\\
	&G(m,k)={\Gamma((k+1),2ma)}/{m^{k+1}}.\label{funG}
\end{align}

%
%
%
 \section{System Model}
In this paper, we consider a highway with RSUs deployed along the road to provide the infrastructure connectivity to VUs. The RSU locations are modeled as a homogeneous $1$D PPP $\Phi=\{\z_{i},\,\forall i \in \mathbb{N}\}$ with density $\lambda_{\rsu}$. The RSUs are assumed to transmit at power $P_{\rm t}$. \\
\textbf{Modeling of traffic scenarios:} As shown in Fig. \ref{N-PTS}, we have considered two types of vehicular traffic models (i) PTS and (ii) N-PTS. The spatial models under these traffic settings are described next. \\
\textbf{(i) Modeling of VU locations in PTS:}
The VU locations in PTS are modeled as $1$D MCP \cite{pandey2023fundamentals} denoted as $\Phi_{\mrm}$. In a $1$D MCP, the parent PP is a $1$D PPP $\Phi_{\prm}=\{\bm{x}_{i}\}$ with density $\lambdaPT$. For each parent point $\bm{x}_{i} \in \Phi_{\prm}$, an independent and identically distributed daughter PP $\Phi_{\bm{x}_{i}}=\{\bm{y}_{i,j}\}$ is generated in $1$D ball $\bt_{1}(\bm{x}_{i},a)$, where $\bm{y}_{i,j}\in \Phi_{\bm{x}_{i}}$ denotes the location of $j$th  daughter point of $i$th parent point. Further, in each daughter PP,  the number of  its daughter points is Poisson distributed with parameter $m$ hence the density of a daughter PP is $\lambda_{\drm}(\bm{y})=\frac{m}{2a}\1(\|\bm{y}-\bm{x}\|\leq a)$. Finally, the union of all the daughter PPs gives the $1$D MCP $\Phi_{\mrm}$ \ie
\begin{align}
	\Phi_{\mrm}=\bigcup\nolimits_{\bm{x}_{i}\in \Phi_{\prm}}\Phi_{\bm{x}_{i}}.
\end{align}The effective vehicular density is $\lambda=m \lambdaPT$.\\
\textbf{(ii) Modeling of VU locations in N-PTS:} As a realistic baseline for benchmarking the performance of PTS, we also consider N-PTS under similar general assumptions in this paper. The locations of the VUs moving independently in N-PTS is modeled as  $1$D PPP \cite{1dpaper} $\Phi_{v}=\{\bm{v}_{i}\}$ with density $\lambda$. We examine V2V and V2I communications for a VU under both PTS and N-PTS.\\
{\textbf{V2V communication:}
VUs must continually exchange messages that report changes in speed, location, and distance to nearby vehicles, regardless of traffic conditions. We assume that two VUs can communicate successfully if they are within the communication range \( R_{\rm b} \).\\ 
\textbf{V2I communication and  association region of RSU:} 
For V2I communication, a VU connects to the cellular network through the nearest RSU. Since all RSUs use the same transmit power \( P_{t} \), each RSU's coverage area forms a Voronoi cell. For any RSU located at \( \bm{z} \), its corresponding $1$D Voronoi region is denoted by \( V_{\bm{z}} \) given as

\begin{equation}\label{voronoi-region}
	V_{\bm{z}}=\{\X \in \mathbb{R}:|\X-{\bm{z}}|\leq |\X-{\bm{z}}_{j}|,\,\forall\, {\bm{z}}_{j} \in \Phi\setminus \{\bm{z} \}.
\end{equation}Before proceeding, we define the typical and tagged RSUs.
\begin{figure}[t!]
		\centering
	(a)\includegraphics[width=.8\linewidth]{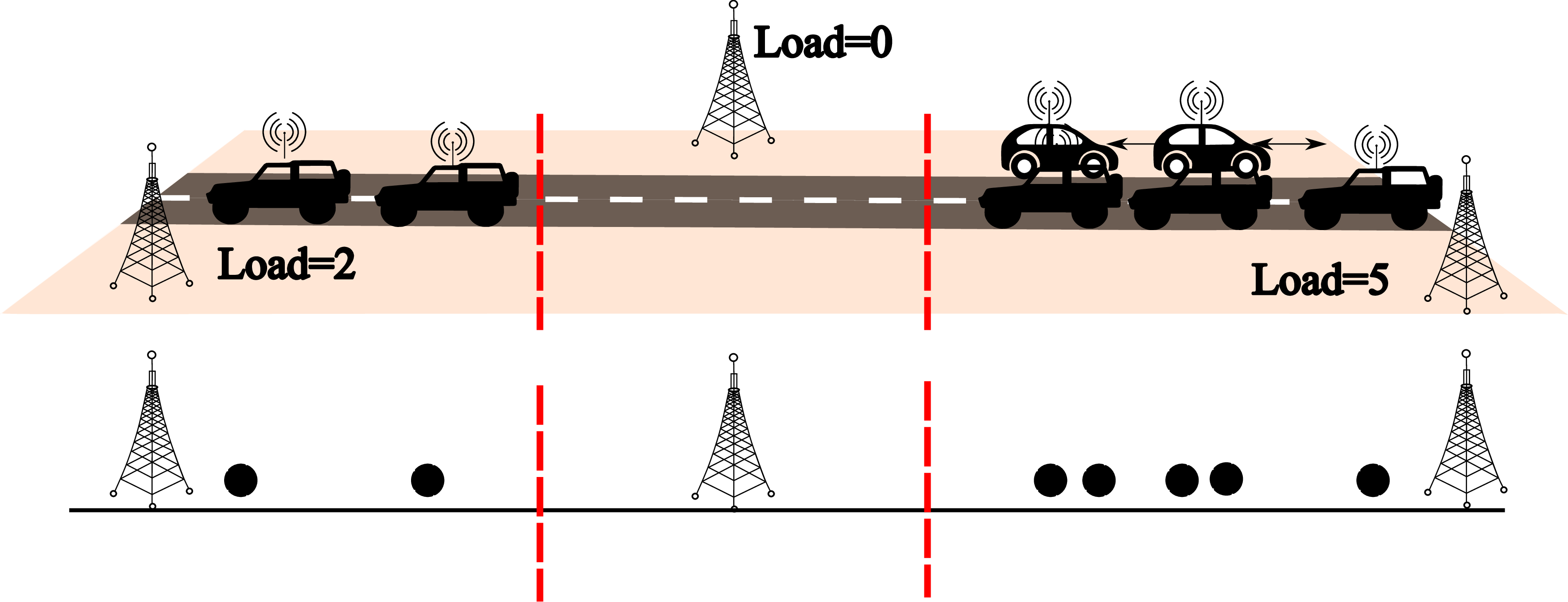}
	(b)\includegraphics[width=.8\linewidth]{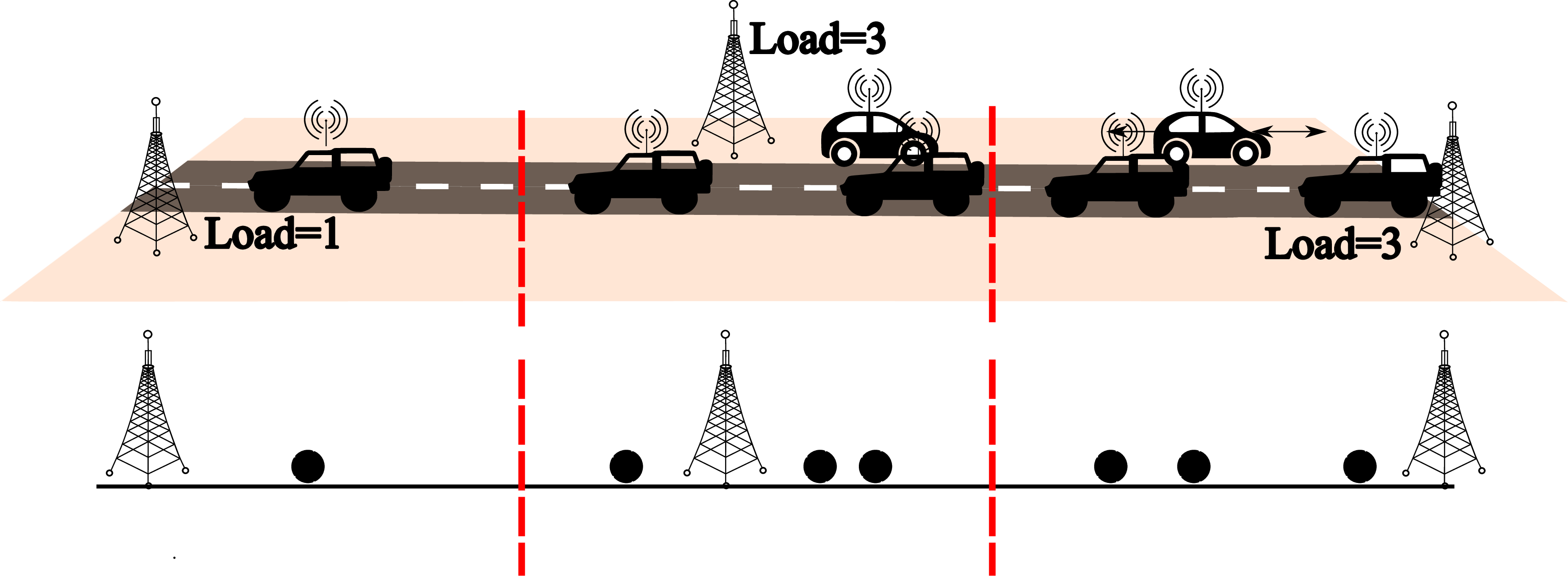}
	\caption{An illustration showing a vehicular network for (a) PTS and (b) N-PTS,
	top figure shows the RSUs and vehicles, while bottom models VUs as points in a PP. Dotted lines mark the boundaries for the serving region of RSUs.}\label{N-PTS}
\end{figure}
 \begin{definition}[The typical RSU and its association region]
The concept of typicality is used in PPs to quantify the typical (or average) properties of the process as observed from a specific location \cite{SGBook2022}. In stationary settings, such as the ones considered here, it is sufficient to place the typical point at the origin. The 1D Voronoi cell associated with the typical RSU (the one at the origin) is its association (serving) region and is termed the typical cell.  
\end{definition}
 \begin{definition}[Tagged RSU and associated cell]
 	An RSU that serves the typical vehicle is defined as the tagged RSU and its association region is termed the tagged cell.
\end{definition} The PDFs for length $L$ of the typical and length $L_{\ob}$ of tagged cell are respectively given as \cite[Lemma-1]{tanemura2003statistical,tagged1d} 
 \begin{align}\label{length-distribution}
f_{L}(l)=4\lambda_{\rsu}^{2}l\exp{\left(-2\lambda_{\rsu}l\right)},\,f_{L_{\ob}}(l_{\ob})=4\lambda_{\rsu}^{3}l_{\ob}^{2}e^{-2\lambda_{\rsu}l_{\ob}}.
\end{align}
}
{To compute the mean and variance of the load on the typical and tagged cells, we use the moment generating functions (MGFs) of \( L \) and \( L_{\ob} \), which is obtained directly from their PDFs.
	 
\begin{lemma}\label{MGF}
	The MGF of $L$ and $L_{\ob}$ is respectively
$		\mathcal{M}_{L}(t)={4\lambda_{\rsu}^{2}}{(t-2\lambda_{\rsu})^{-2}},\,\mathcal{M}_{L_{\ob}}(t)={8\lambda_{\rsu}^3}{\left(2\lambda_{\rsu}-t\right)^{-3}}.$
\end{lemma}
Since the VU locations under PTS are modeled as a $1$D MCP, we first present key results for the 1D MCP before analyzing the load distribution for PTS.
}
\section{Distributional Properties of $1$D MCP}
First, we present distribution of the number ($S(r)$) of points of a $1$D MCP falling in a ball of radius $r$ with its center at any arbitrary point in $\R$. The PGF of $S(r)$ is given as \cite{pandeykth}
\begin{align}\label{PGFMCP}
	&\mathcal{P}_{S(r)}(s)=e^{g(s,r)},\text{ where }\\
	&g(s,r)=2\lambdaPT\left(\int_{0}^{r+a}e^{m\bar{\A}(r,a,x)(s-1)}\dv x-(r+a) \right)\nonumber\\
	&=2\lambdaPT\left[|r-a|e^{m\bar{\beta}(r)(s-1)}-\!(r+a)+\!\frac{e^{m(s-1)\bar{\beta}(r)}-1}{(m/(2a))(s-1)}\right].\label{g(s,t)}
\end{align}Note that PMF can be obtained from PGF via the relation
\begin{align}\label{PMF}
	&\mathbb{P}\left[S(r)=k\right]=p_{S(r)}(k)=\left[{\mathcal{P}^{(k)}_{S(r)}(s)}/{k!}\right]_{s=0}\nonumber\\
	&\stackrel{(a)}=\frac{e^{g(0,r)}}{k!}\sum\nolimits_{\mathrm{B}_{k}}\frac{k!}{b_1!\ldots b_{k}!}\prod\nolimits_{1\leq i\leq k}
	\left({g^{(i)}(0,r)}/{i!}\right)^{b_i}\!\!,
\end{align}where step $(a)$ is obtained using the Faà di Bruno's lemma \cite{faadibruno}, the sum is over the set $\mathrm{B}_{k}$ consisting of all $k$-tuples $\{b_{1}, b_{2},\ldots,b_{k}\}$ with $b_{i} \geq0$ and $b_{1}+2b_{2}+\ldots+
kb_{k} = k$
and 
\begin{align}
\left[g^{{(i)}}(s,r)\right]_{s=0}\!\!&=\int_{0}^{r+a}\!\!\!\!\!2\lambdaPT(m\bar{\A}(r,a,x))^{i}e^{-m\bar{\A}(r,a,x)}\dv x\\
g^{(i)}(0,r)&=2\lambdaPT\left[\left(m\bar{\beta}(r)\right)^{i}e^{-m\bar{\beta}(r)}|r-a|+\right.\nonumber\\
&\left.{\left(\Gamma(k+1)-\Gamma_{\rm u}(k+1,m\bar{\beta}(r))\right)}/{(m/(2a))}\right].\nonumber
\end{align} Note that, $g^{(0)}(s,r)=g(s,r)$. Further, to find the variance and skewness of $S(r)$, we need  $\kappa(r,k)\stackrel{\Delta}=\lim_{s\rightarrow1}g^{(k)}(s,r)$. Also note that $\lim_{s\rightarrow1}g(s,r)=0$. Further,
\begin{align}
&\kappa(r,k)=
2\lambdaPT\left(m\bar{\beta}(r)\right)^k\left[r+a-\beta(r){k}/{(k+1)}\right].\label{kappadefin}\\
&=\begin{cases}\label{kappadefin2}
	2\lambdaPT (m r/a)^{k}\left(a+{r(1-k)}/{(1+k)}\right),& \text{if } r<a,\\
	2\lambdaPT m^{k}\left(r+ {a(1-k)}/{(1+k)}\right),              &\text{if } a<r.
\end{cases}
\end{align}
Now, we present the results frequently used in the subsequent analysis (for the proof, see Appendix \ref{proof_of_lemma1}).
	\begin{align}
		&I(n,k)\!=\!\int_{0}^{\infty}\!\!\!\!	\kappa^{n}(r/2,k)f_{L}(r)\dv r=I_{1}(n,k)+I_{2}(n,k),\label{i1nk}\\
		&\widetilde{I}(n,k)\!=\!\!\int_{0}^{\infty}\!\!\!\!\!	\kappa^{n}(r/2,k)f_{L_{\ob}}(r)\dv r=\widetilde{I}_{1}(n,k)+\widetilde{I}_{2}(n,k),\label{i2nk}\\
&\text{where, } I_{1}(n,k)=(2\lambdaPT)^{n}(m/(2a))^{nk}a^{n}\sum\nolimits_{j=0}^{n}{n \choose j}\nonumber\\
&\left(\frac{ (1-k)}{2a(1+k)}\right)^{j}\frac{\left[\Gamma(nk+j+2)-\Gamma(nk+j+2,4\lambda_{\rsu}a)\right]}{(2\lambda_{\rsu})^{nk+j}},\nonumber\\
&I_{2}(n,k)=(2\lambda_{\rsu})^{2}\left(2\lambdaPT m^{k}\right)^{n}\eta^{n}_{1}\sum_{j=0}^{n}\frac{{n \choose j}}{(2\eta_{1})^{j}}\frac{\Gamma((j+2),4\lambda_{\rsu}a)}{(2\lambda_{\rsu})^{j+2}},\nonumber\\
&\widetilde{I}_{1}(n,k)=\lambda_{\rsu}(2\lambdaPT)^{n}(m/(2a))^{nk}a^{n}\sum\nolimits_{j=0}^{n}{n \choose j}\nonumber\\
&\left(\frac{ (1-k)}{2a(1+k)}\right)^{j} \frac{\left[\Gamma(nk+j+3)-\Gamma(nk+j+3,4\lambda_{\rsu}a)\right]}{(2\lambda_{\rsu})^{nk+j}},\nonumber\\
&\widetilde{I}_{2}(n,k)=4\lambda_{\rsu}^{3}\left(2\lambdaPT m^{k}\eta_{1}\right)^{n}\sum\nolimits_{j=0}^{n}\frac{{n \choose j}}{(2\eta_{1})^{j}}\times\nonumber\\
&\quad\quad\quad\quad\frac{\Gamma((j+3),4\lambda_{\rsu}a)}{(2\lambda_{\rsu})^{j+3}},\nonumber
\end{align}with $\eta_{1}=(a(1-k))/(1+k)$. Using these results, we derive the load distribution for the typical and tagged RSUs under PTS and N-PTS in the next section.
{\section{ Load Distribution on the Typical RSU}
Let $S_{\pt}$ (or $S_{\Nrm}$) denote the number of VUs falling in the serving region of the typical RSU for PTS (or N-PTS). In this section, we derive their distribution.
\subsection{PTS}
The distribution of $S_{\pt}$ for PTS is given in the following theorem. 
\begin{theorem}\label{theorem1}
	The PGF of the VU's load $S_{\pt}$ on the typical RSU in PTS is 
	\begin{align}
		\mathcal{P}_{{S_{\pt}}}(s)=\int_{0}^{\infty}\mathcal{P}_{S(t/2)}(s)f_{L}(t)\dv t,
	\end{align}
where $\mathcal{P}_{S(t/2)}(s)$ is given in \eqref{PGFMCP}. Further, the PMF of $S_{\pt}$ is
\begin{align}\label{sp-eq}
&\prob{S_{\pt}=k}=p_{S_{\pt}}(k)=\int_{t=0}^{\infty}\sum\nolimits_{\mathrm{B}_{k}}\frac{e^{g\left(0,t/2\right)}}{b_{1}! \ldots b_{k}!}\nonumber\\
&\prod\nolimits_{1\leq i\leq k}\left({g^{(i)}(0,t/2)}/{i!}\right)^{b_{i}}f_{L}(t)\dv t.
\end{align}
\end{theorem}
\begin{proof}
 Recall that the typical RSU is placed at the origin. Conditioned on the length $L$ of the serving region, the number of VUs in PTS falling inside $\bt_{1}(\ob,L/2)$ is provided in \eqref{PGFMCP}. Finally, deconditioning using the PDF of $L$, we get the PGF of $S_{\pt}$. 
\end{proof}Now, we present the mean, variance, and skewness of \( S_{\pt} \), which require the first, second, and third derivatives of its PGF evaluated at \( s = 1 \). These are obtained from \eqref{i1nk}, and then using \eqref{variance} and \eqref{skew}, along with results from Lemma~\ref{MGF}. We omit the full proof due to lack of space.

\begin{corollary}
The mean, variance and the third moment of $S_{\pt}$ are given as (for proof, see Appendix \ref{proof_of_cor_2})
\begin{align}
	&\mathbb{E}\left[S_{\pt}\right]={m\lambdaPT}/{\lambda_{\rsu}},\,\mathrm{Var}\left[S_{\pt}\right]={m\lambda_{\pt}}/{\lambda_{\rsu}}-\left({m\lambda_{\pt}}/{\lambda_{\rsu}}\right)^2\nonumber\\
	 &+I(2,1)+I(1,2),\,\mathbb{E}\left[S_{\pt}^{3}\right]=I(3,1)+I(1,3)\nonumber\\
	&+3I(1,2)(I(1,1)+1)+3I(2,1)+{m\lambdaPT}/{\lambda_{\rsu}},\nonumber
\end{align}where $I(n,k)$ is given in \eqref{i1nk}. Using mean, variance and the third moment, we can easily obtain the skewness $\mathbb{S}[S_{\pt}]$ of $S_{\pt}$  using \eqref{skew}.
\end{corollary}}
\subsection{N-PTS}
Now, we present the PGF and PMF of the VU's load under N-PTS. 
\begin{theorem}\label{theorem2}
	The PGF and PMF of load on the typical RSU in N-PTS are given as
		\begin{align*}
		&\mathcal{P}_{{S_{\Nrm}}}(s)={\left(1-\frac{\lambda (s-1)}{(2 \lambda_\rsu)}\right)^{-2}},\,	p_{S_{\Nrm}}(k)=\frac{4\lambda_{\rsu}^2\lambda^{k}(k+1)}{\left(\lambda+2\lambda_{\rsu}\right)^{k+2}}.
	\end{align*} 
\end{theorem}
\begin{proof}
		Conditioned on length of the typical Voronoi cell $L=l$, the number of VU in length $l$ is Poisson distributed with mean $\lambda l$. Hence, the PGF is $\mathcal{P}_{S_{\Nrm}\vert L=l}(s,l)=\exp\left(\lambda l (s-1)\right)$.
		Deconditioning using the PDF of $L$, we get the PGF. Using the PGF's $k$th derivative and simplifying further, we get the PMF of $S_{\Nrm}$.
\end{proof}
The mean, variance, and skewness of $S_{\Nrm}$ can be determined using a method similar to that used for $S_{\pt}$.
\newcommand{\ut}{\gamma_{\rsu}}
\begin{corollary}\label{Cor:2.1}
The mean $\gamma_{\rsu}$, variance and skewness of $S_{\Nrm}$ are	\begin{align}
	&\gamma_{\rsu}=\mathbb{E}\left[S_{\Nrm}\right]=\lambda \mathbb{E}\left[L\right]={\lambda}/{\lambda_{\rsu}},\,\mathrm{Var}\left[S_{\Nrm}\right]={\gamma_{\rsu}^2/}{2}+\gamma_{\rsu},\\
		&\mathbb{S}[S_{\Nrm}]={\left(({1}/{2})\ut^3+({3}/{2})\ut^2+\ut\right)}{\left(({1}/{2})\ut^2+\ut\right)^{-3/2}}\nonumber.
\end{align} 
 \end{corollary}
{
\section{Load Distribution on the Tagged RSU}
Now, we present vehicular load distribution on the tagged RSU. Let the load on the tagged RSU in PTS and N-PTS is $\widetilde{S}_{\pt}$ and $\widetilde{S}_{\Nrm}$, respectively. 
\subsection{PTS}
{In PTS, the load on the tagged RSU arises from (i) the tagged platoon and (ii) the remaining platoons on the road. Without loss of generality, let the typical VU of \( \Phi_{\mrm} \) is at the origin. Conditioned on this location, the process \( \Phi_{\mrm} \) can be represented as the union of two independent PPs. The first, \( \Phi'_{\mrm} \), is an independent copy of \( \Phi_{\mrm} \). The second, \( \Phi_{\bm{x}_{0}} \), is the daughter PP associated with the typical VU, whose parent is located at \( \bm{x}_{\ob} \). Let \( L_{\ob} \) denote the length of the tagged serving. Conditioned on \( L_{\ob} \), the load \( \widetilde{S}_{\pt} \) on the tagged RSU is the sum of two independent RVs. Specifically, the load \( \widetilde{S}_{\rm P} \) on a tagged RSU of length \( L_{\ob} \) is given by
	\begin{align}
\widetilde{S}_{\rm P}&=\Phi^{'}_{\mrm}\left(\bt_{1}(\ob,L_{\ob}/2)\right)+\Phi_{\bm{x}_{\ob}}(\bt_{1}(\ob,a)\cap L_{\ob})\nonumber\\
&=S(L_{\ob}/2)+V_{\rm m}(L_{\ob}/2),\label{eq:Sptilde}
\end{align}where $V_{\rm m}(\cdot)$ and $S(\cdot)$ denote the load due to the tagged platoon and the remaining platoons in a 1D ball of length $L_{\ob}$, respectively. Note that the total load on the tagged RSU counting the typical vehicle is $\widetilde{S}_{\pt}+1$.
Now, conditioning on \( L_{\ob} = t \), we derive the PGF, mean, and variance of \( V_{\rm m}(t/2) \). Using these results, we then compute the mean and variance of \( \widetilde{S}_{\pt}\).
}
\begin{theorem}\label{thm:3}
	The PGF $\mathcal{P}_{V_{\mrm}(t/2)}(s)$ of the load $V_{\rm m}(t/2)$ on the tagged RSU with cell length $t$ due to tagged platoon is (for proof, see Appendix \ref{proofofpmf})
	\begin{align*}
		&\mathcal{P}_{V_{\mrm}(t/2)}(s)=\frac{1}{at}\int_{0}^{a+t/2}{e^{m\bar{\A}(t/2,a,y)(s-1)}}f_{y}(y)\dv y\\
		&=\begin{cases}
			\frac{(a-t/2)e^{\frac{mt(s-1)}{2a}}}{a}+\frac{2e^{\frac{mt(s-1)}{2a}}}{m(s-1)}\\
			+\frac{4a(1-e^{\frac{mt(s-1)}{2a}})}{m^2t(s-1)^{2}},&\text{if }a>t/2,\\
			\frac{2}{t}e^{m(s-1)}(t/2-a)+\frac{4a e^{m(s-1)}}{mt(s-1)}\\
			-\frac{4ae^{m(s-1)}}{t m^{2}(s-1)^{2}}+\frac{4a}{m^{2}t(s-1)^{2}},&\text{if }t/2>a.
		\end{cases}
	\end{align*}
	where $f_{y}(y)$ is given as
	\begin{align}\label{distr_fy}
		f_{y}(x) =\begin{cases}
			t u(x)-\mathcal{R}(x-(a-t/2))\\
			+\mathcal{R}(x-(a+t/2)), & \text{if } a > t/2, \\
			2au(x)-\mathcal{R}(x-(t/2-a))\\
			+\mathcal{R}(x-(a+t/2)), & \text{if } a<t/2.
		\end{cases}
	\end{align}
\end{theorem}Now, we present the PMF, mean and the variance of $V_{\mrm}(t/2)$ using the PGF of $V_{\mrm}(t/2)$.
\begin{corollary}
The PMF of  $V_{\mrm}(t/2)$ is
	\begin{align}\label{vnt}
		&\prob{V_{\mrm}(t/2)=n}=\nu_{n}(t/2)=	\left[{\mathcal{P}^{(n)}_{V_{\mrm}}(s,t/2)}/{n!}\right]_{s=0}\nonumber\\
		&\!\!\!\!\!\!=\frac{1}{n!}\begin{cases}
			\frac{(a-t/2)((mt)/2a)^{n}e^{-mt/(2a)}}{a}	+\frac{4a(n+1)!}{ m^{2}t}	-\frac{2}{m}\sum_{k=0}^{n}{n\choose k}\\
			(mt/(2a))^{n-k}e^{-mt/(2a)}k!\left(1+\frac{k+1}{t}\right),\quad\text{if } a>t/2, \\
			\frac{2}{t}(\frac{t}{2}-a)m^{n}e^{-m}+\frac{4a}{m^{2}t}(n+1)!-\frac{4a}{mt}\sum_{k=0}^{n}{n \choose k}\\
			m^{n-k}e^{-m}k!\left(1+{(k+1)}/{m}\right),\quad\text{if } a<t/2. 
		\end{cases}
	\end{align}
\end{corollary}
\begin{corollary}
	The mean and the variance of $V_{\mrm}(t/2)$ for a given length $t$ of the tagged platoon are (for proof, see Appendix \ref{proofofthirdmoment})
	\begin{align*}
		&\expect{V_{\mrm}(t/2)}=
		\begin{cases}
			 \left(1-\frac{t}{2a}\right)\frac{m t}{2a}+\frac{m t^{2}}{12a^{2}},&\text{if }  t<2a,\\
			  \frac{2m(t/2-a)}{t}+\frac{4ma}{3t},&\text{if } t>2a.
		\end{cases}\\
&\mathrm{Var}\left[V_{\mrm}(t/2)\right]=
\begin{cases}
\frac{5m^{2}t^{3}}{48a^{3}}+\frac{m t}{2a}-\frac{mt^{2}}{6a^{2}}-\frac{m^{2}t^{4}}{36a^{4}}, &\text{if } t<2a,\\
m^{2}\left(1-{a}/{t}\right),&\text{if } t>2a.
\end{cases}
\end{align*}
\end{corollary}Deconditioning the mean and variance of $V_{\rm m}(t/2)$ using the distribution of the tagged cell length $L_{\ob}$ in \eqref{length-distribution}, we obtain the following lemma.
\begin{lemma}
	The mean and variance of the number of VUs in the tagged platoon are given as
	\begin{align}
&\expect{V_{\mrm}}=\lambda_{\drm} 4 \lambda_{\rsu}^{3}F(2\lambda_{\rsu},3)-{4\lambda_{\rsu}^{3}\lambda_{\drm}F(2\lambda_{\rsu},4)}/{(3a)}\nonumber\\
&+4m\lambda_{\rsu}^{3}G(2\lambda_{\rsu},2)-{8ma\lambda_{\rsu}^{3}G(2\lambda_{\rsu},1)}/{3},\label{meanvm}\\
&\mathrm{Var}[V_{\mrm}]={5\lambda_{\drm}^{2}\lambda_{\rsu}^{3}F(2\lambda_{\rsu},5)}/{(3a)}+4\lambda_{\rsu}^{3}\lambda_{\drm}F(2\lambda_{\rsu},3)\nonumber\\
&-{\lambda_{\drm}4\lambda_{\rsu}^{3}F(2\lambda_{\rsu},4)}/{(3a)}-{4\lambda_{\rsu}^{3}\lambda_{\drm}^{2}F(2\lambda_{\rsu},6)}/{(9a)}\nonumber\\
&+m^{2}4\lambda_{\rsu}^{3}\left(G(2\lambda_{\rsu},2)-G(2\lambda_{\rsu},1)\right),\label{varvm}
	\end{align}
	where $F(\cdot,\cdot)$ and $G(\cdot,\cdot)$ is provided in \eqref{funF} and \eqref{funG}.
\end{lemma}
%
Now, we have the PGF of \( S \) from \eqref{PGFMCP} and the PGF of \( V_{\rm m}(t/2) \) from Theorem~\ref{thm:3}. Their product gives the PGF of \( \widetilde{S}_{\pt} \), stated in the following theorem.

\begin{theorem}
	The PGF for load $\widetilde{S}_{\pt}$ on the tagged RSU is (for proof, see Appendix \ref{proof_of_theorem3})
	\begin{align*}
		&\mathcal{P}_{\widetilde{S}_{\pt}}(s)=	\int_{t=0}^{\infty}\mathcal{P}_{S(t/2)}(s)\mathcal{P}_{V_{\mrm}(t/2)}(s) f_{L_{\ob}}(t)\dv t,
	\end{align*}where $V_{\mrm}(t/2)$ denotes the number of VUs of the tagged platoon, falling in the tagged cell which depends on length $t$ of the segment of intersection between the tagged platoon and the tagged cell. 
\end{theorem}Now, using the PGF of $\widetilde{S}_{\pt}$ and the PMFs of $V_{m}(t/2)$ and $S(t/2)$, we derive the PMF of  $\widetilde{S}_{\prm}$, as given in the following Theorem.
\begin{theorem}
	The PMF $p_{\widetilde{S}_{\pt}}(k)$ of $\widetilde{S}_{\pt}$ is
	\[p_{\widetilde{S}_{\pt}}(k)=\frac{2}{k!}\int_{t=0}^{\infty}\sum\nolimits_{n=0}^{k} {k\choose n}p_{S(t)}(n-k)\nonumber\\
	\nu_{n}(t) f_{L_{\ob}}(2t)\dv t, \]
	where $p_{S(\cdot)}(k)$ is given in \eqref{PMF}
	and $\nu_{n}(t)$ is given in \eqref{vnt}.
\end{theorem}Now, we present the mean, variance, and third moment of $\widetilde{S}_{\pt}$, which can be easily derived using its PGF. To save space, we have skipped the proof. 
\begin{corollary}
	The mean, variance and the third moment of $\widetilde{S}_{\pt}$ are
	\begin{align}
		&\mathbb{E}\left[\widetilde{S}_{\pt}\right]={(3m\lambdaPT)}/{(2\lambda_{\rsu})}+\mathbb{E}\left[V_{\mrm}\right],\label{mean_tagged}\\
		&\mathrm{Var}\left[\widetilde{S}_{\pt}\right]=\mathrm{Var}\left[V_{\mrm}\right]+{(3m\lambda_{\pt})}/{(2\lambda_{\rsu})}-\left({(3m\lambda_{\pt})}/{(2\lambda_{\rsu})}\right)^2\nonumber\\
		&+\widetilde{I}(2,1)+\widetilde{I}(1,2),\label{variance_tagged}\\
		&\mathbb{E}\left[\widetilde{S}_{\pt}^{3}\right]=3\bigg(\mathrm{Var}[\widetilde{S}_{\pt}]+\left(\mathbb{E}[\widetilde{S}_{\pt}]\right)^{2}\bigg) -2\mathbb{E}[\widetilde{S}_{\pt}]+\!\!\int_{0}^{\infty}\!\!\!\!\!\left(f_{3}(t)+\right.\nonumber\\
		&\left.3f_{2}(t)\nu_{1}(t/2)+3f_{1}(t)\nu_{2}(t/2)+\nu_{3}(t/2)\right)f_{L_{\ob}=t}(t)\dv t,\label{third_mean_tagged}
	\end{align}
	where $\widetilde{I}(\cdot,\cdot),\,\nu_{n}(\cdot)$ are given in \eqref{i2nk} and \eqref{vnt} respectively, 
and
	\begin{align}
		&f_{1}(r)=\kappa(r/2,1),\,f_{2}(r)=\kappa^{2}(r/2,1)+\kappa(r/2,2),\\
		&f_{3}(r)=\kappa^{3}(r/2,1)+3\kappa(r/2,2)\kappa(r/2,1)+\kappa(r/2,3),
	\end{align}with $\kappa(\cdot,\cdot)$ given in \eqref{kappadefin}.
\end{corollary}
}
{\subsection{N-PTS}
\begin{theorem} \label{Theorem3}
	The PGF $\mathcal{P}_{{\widetilde{S}_{\Nrm}}}(s)$ of $\widetilde{S}_{\Nrm}$ is 
	\begin{align*}
	&\mathcal{P}_{{\widetilde{S}_{\Nrm}}}(s)=\int_{0}^{\infty}e^{\lambda l_{\ob}(s-1)} f_{L_{\ob}=l_{\ob}}(l_{\ob})\dv l_{\ob}={\left(1-{.5\gamma_{\rsu}}(s-1)\right)^{-3}},
	\end{align*}
The PMF of  $\widetilde{S}_{\Nrm}$ is $p_{\widetilde{S}_{\Nrm}}(k)=\left(.5\gamma_{\rsu}\right)^{k}\frac{.5(k+2)(k+1)}{\left(1+.5\gamma_{\rsu}\right)^{3+k}}$, recall that $\gamma_{\rsu}={\lambda}/{\lambda_{\rsu}}$.
\end{theorem}
\begin{proof}
	For the PGF of $\widetilde{S}_{\Nrm}$, we assume that the typical point $\bm{v}_{\ob}$ is located at the origin. From the Slivnyak's theorem \cite{SGBook2022}, $\{\Phi_{v}\setminus\{\bm{v}_{\ob}\}\}=\{\Phi_{v}\}$. Hence, conditioned on $L_{\ob}=l_{\ob}$, the number of VUs of $\Phi_{v}\setminus\{\bm{v}_{\ob}\}$ inside the tagged cell is Poisson distributed with mean $\lambda l_{\ob}$. Finally, deconditioning using PDF of $L_{\ob}$, we get the PGF of $\widetilde{S}_{\Nrm}$. Using the PGF, we can easily obtain the PMF of $\widetilde{S}_{\Nrm}$.
\end{proof}
\begin{corollary}
	The mean, variance and the third moment of $\widetilde{S}_{\Nrm}$ are
	\begin{align}
&\mathbb{E}\left[\widetilde{S}_{\Nrm}\right]=1.5\gamma_{\rsu},\,\mathrm{Var}\left[\widetilde{S}_{\Nrm}\right]=3\left(.5\gamma_{\rsu}\right)^{2}+1.5\gamma_{\rsu},\label{meanvariance}\\
&\mathbb{E}[\widetilde{S}_{\Nrm}^{3}]=7.5\gamma_{\rsu}^{3}+9\gamma_{\rsu}^{2}+1.5\gamma_{\rsu},\label{thirdmean}
	\end{align}
where $\gamma_{\rsu}=\lambda/\lambda_{\rsu}$.
Using \eqref{meanvariance}, \eqref{thirdmean} in \eqref{skew}, we can get the skewness of $\widetilde{S}_{\Nrm}$.
\end{corollary}
\section{Vehicular Connectivity Analysis} In this section, we analyze the performance of vehicular communication in PTS and N-PTS. 
\subsection{V2V connectivity}
This section analyzes the V2V communication performance of VUs under PTS and N-PTS. As noted earlier, two VUs can communicate if they lie within a fixed distance—called the communication radius, \( R_{\brm} \)—of each other. Consider the typical VU located at the origin in both the PTS and N-PTS cases. We use \emph{connectivity degree} as the performance metric, defined as the number of VUs within communication range of the typical VU. This metric reflects the level of V2V connectivity in the network. Let \( \widetilde{N}_{\pt} \) and \( \widetilde{N}_{\Nrm} \) denote the connectivity degree under PTS and N-PTS, respectively. The following theorem presents the distribution of the connectivity degree.

  \begin{theorem}\label{Theoremnodedegree}
  	The PGF for number of VUs falling in the communication range of the typical VU is (for proof, see Appendix \ref{proof_of_nodedegree})
  	\begin{align}
  		&\mathcal{P}_{\widetilde{N}_{\pt}}\left(s,R_{\brm}\right)=\mathcal{P}_{S}\left(s,R_{\brm}\right)\frac{\int_{0}^{a}e^{m\bar{\A}(R_{\brm},a,x)(s-1)}\dv x}{a}.\label{secondterm}\\	&\mathcal{P}_{\widetilde{N}_{\Nrm}}\left(s,R_{\brm}\right)=\exp\left(\lambda R_{\brm}(s-1)\right).
  	\end{align}The PMF of $\widetilde{N}_{\Nrm}$ is $	\mathbb{P}\left[\widetilde{N}_{\Nrm}=k\right]=\frac{1}{k!}\left(\lambda R_{\brm}\right)^{k}\expS{-\lambda R_{\brm}},\,\forall k\geq 0$, and PMF for $\widetilde{N}_{\pt}$ is
  	\begin{align*}
  		&p_{\widetilde{N}_{\pt}}(k)=\frac{1}{k!}\sum\nolimits_{n=0}^{k} {k\choose n}(k-n)! 	p_{S(R_{\rm b}/2)}(k-n)\frac{1}{a}\nonumber\\
  		&\int_{0}^{a}\left(m\bar{\A}(R_{\brm}/2,a,x)\right)^{n}e^{-m\bar{\A}(R_{\brm}/2,a,x)}\dv x,\quad  k\geq0.
  	\end{align*}where $p_{S(R_{\rm b}/2)}(\cdot)$ is given in \eqref{PMF}.
\end{theorem}}The second term in \eqref{secondterm} arises from the platoon associated with the typical VU, reflecting its correlation with the presence of the typical VU.

\subsection{Cellular coverage}
We now analyze the coverage provided by the RSUs under PTS and N-PTS. Note that if RSU is not serving any VUs, it is considered inactive and doesn't contribute to the interference. 
Therefore, we first compute the density of active RSUs as it directly determines the sum interference. This study provides critical insights into the relationship between vehicle mobility patterns and the
utilization of RSUs in the network. Let us define the active probability  as the probability that the typical RSU is serving atleast one VU. Let  $\prm_{\prm}$ and $\prm_{\mathrm{n}}$ denote the active probability under PTS and N-PTS, given as
\[\prm_{\prm}=1-p_{S_{\pt}}(0),\quad \prm_{\mathrm{n}}=1-p_{S_{\Nrm}}(0). \]
We can easily derive the active probability using the results presented in Theorem \ref{theorem1} and Theorem \ref{theorem2}. 
\begin{theorem}\label{thm:8}
The probabilities that an RSU is active under PTS and N-PTS are given, respectively, as
	\begin{align*}
&\prm_{\prm}\!=\!\!1-\!16\lambda_{\rsu}^{2}\int_{0}^{\infty}\!\!\!\!\! e^{g(0,t)}f_{L}(2t)\dv t,\,\,\prm_{\mathrm{n}}={(1+4 \gamma_{\rsu})}{(1+2\gamma_{\rsu})^{-2}},\!\!\!
	\end{align*}where $\gamma_{\rsu}={\lambda_{\rsu}}/{\lambda}$ and $ f_{L}(\cdot),\,g(0,t)$ are provided in \eqref{length-distribution}, \eqref{g(s,t)}.
\end{theorem}The density $\lambda_{\rm a}$ of active RSUs is given by $\prm_{\prm}\lambda_{\rsu}$,  $\prm_{\rm n}\lambda_{\rsu}$
under PTS and N-PTS, respectively. Now, due to dependent thinning, the active RSUs don't form a PPP. However for tractability, we approximate the distribution of the active RSUs as PPP $\Phi_{\rm a}$, consistent to previous works \cite{gupta2015sinr}.

Now, we consider the typical VU to determine its CP. Owing to stationarity, we take it at the origin. The association is distance based where each VU is connected to the closest RSU. For each link, we consider Rayleigh fading with pathloss exponent $\alpha$. Then, SINR at the typical VU is 
\begin{align}
	\SINR={H P_{\trm} R^{-\alpha}}/{(I+\sigma^{2})},\label{SINR}
\end{align}where $H$, $P_{\trm}$, $R$, $I$ and $\sigma^{2}$ denotes the Rayleigh fading gain for the serving link, transmit power, the distance of the nearest serving RSU from the typical VU, the interference from the active RSUs and the noise power, respectively. Due to association law, all interfering RSUs are located at a distance more than $r$ \cite{SGBook2022}. The interference due to active RSUs is
\begin{align}
	I=\sum\nolimits_{\z_{i}\in \Phi_{\rm a}\setminus\bt_{1}(\ob,r)} H_{i} P_{\trm} |\z_{i}|^{-\alpha}, \label{interfeence}
\end{align}where $H_i$ denotes the fading gain for the link between $i$th RSU located at ${\z}_i$ and the typical VU. The CP at the typical vehicle is defined as 
\begin{align}\label{CP}
	&\mathrm{p_{c}}(\tau)=\prob{\SINR>\tau}.
\end{align}Since the closest RSU is the serving RSU, the PDF of the distance of the serving RSU is $f_{R}(r)=2\lambda_{\rsu}e^{-2 \lambda_{\rsu}r}.$ The LT of interference at the typical VU due to active RSUs are provided in the following Lemma \cite{SGBook2022}.
\begin{lemma}
	The LT of interference due to active RSUs at the typical vehicle, conditioned on the serving distance to be $r$ is
	\begin{align}
		&\mathcal{I}_{I}(s)=\exp\left(-\frac{2\prm_{\cdot}\lambda_{\rsu} s P_{\trm}r^{1-\alpha}}{\alpha}\right.\nonumber\\
		&\times\left.\frac{{}_2F^{1}\left([1,1-(1/\alpha)];2-(1/\alpha),-sP_{\trm}r^{-\alpha}\right)}{1-(1/\alpha)}\right),
	\end{align}where ${}_2F_1(a, b; c; z)$ denotes a hypergeometric function and $\prm_{\cdot}$  ($\prm_{\prm}$ or $\prm_{\rm n}$ for PTS or N-PTS, respectively) denotes the active probability of RSUs given by Theorem \ref{thm:8}.
	
\end{lemma}
\begin{proof}
From \eqref{interfeence} and the definition of LT \cite{SGBook2022}, we get
\begin{align*}
	&\laplace{I}(s)=\expects{\Phi_{\rm a}}{\expS{-sI}}\\
	&\stackrel{(a)}=\expects{\Phi_{\rm a}}{\expS{-s\sum\nolimits_{\z_{i}\in \Phi_{\rm a}\setminus\bt_{1}(\ob,r)} H_{i} P_{\trm} |\z_{i}|^{-\alpha}}}\\
	&\stackrel{(b)}=\expects{\Phi_{\rm a}}{\prod\nolimits_{\z_{i}\in \Phi_{\rm a}}\frac{1}{1+s P_{\trm}z_{i}^{-\alpha}}},
\end{align*}where step $(a)$ is obtained by replacing $I$ from \eqref{interfeence} and step $(b)$ is using the MGF of RV $H_{i}$. Simplifying further using the PGFL of $1$D PPP, we get the LT of interference as
	\begin{align*}
	\laplace{I}(s)&=\expS{-2\prm_{\cdot}\lambda_{\rsu}\int_{r}^{\infty}\left(1-\frac{1}{1+sP_{\trm}z^{-\alpha}}\right)\dv z}\\
	&\stackrel{(a)}=\expS{-\frac{2\prm_{\cdot}\lambda_{\rsu} s P_{\trm}r^{1-\alpha}}{\alpha}\int_{t=0}^{1}\left(\frac{t^{-1/\alpha}}{1+sP_{\trm}r^{-\alpha}t}\right)\dv t}.
\end{align*}Here step $(a)$ is obtained by replacing $z^{-\alpha}=t r^{-\alpha}\implies\dv z=-\frac{r}{\alpha}t^{-\frac{1}{\alpha}-1}$. Simplifying further, we get the desired result.
\end{proof}Equipped with the LT of $I$, we now present the CP in the following theorem.
\begin{theorem}
	The CP at the typical VU is
	\begin{align*}
		\mathrm{p_{c}}(\tau)=2\lambda_{\rsu}\int_{r=0}^{\infty}\laplace{I}\left({\tau r^{\alpha}}/{P_{\mathrm{t}}}\right)e^{-\frac{\tau r^{\alpha}}{\SNR}-2\lambda_{\rsu}r}\dv r,
	\end{align*}
where $\SNR={P_{\trm}}/{\sigma^{2}}$.
\end{theorem}
\begin{proof}
	From the definition of the CP ($\prob{\SINR>\tau}$) 
		\begin{align*}
	\mathrm{p_{c}}(\tau)&=\prob{\frac{HP_{\trm}R^{-\alpha}}{I+\sigma^{2}}>\tau}\stackrel{(a)}=\prob{H>\frac{\tau(I+\sigma^{2})R^{\alpha}}{P_{\trm}}}\\
	&\stackrel{(b)}=\expects{R}{\expects{\Phi_{\rm a}\vert R}{\prod\nolimits_{\z_{i}\in \Phi_{\rm a}}\frac{1}{1+{\tau R^{\alpha}}|\z_{i}|^{-\alpha}}}e^{-\frac{\tau R^{\alpha}}{\SNR}}},
\end{align*}where step $(a)$ follows from \eqref{SINR}, and step $(b)$ uses the MGF of the RV $H_{i}$. Using the PGFL of PPP $\Phi_{\rm a}$ and deconditioning using the distribution of $R$, we get the CP.
\end{proof}
\textbf{MD of the CP:}
For a given realization, the CP is a RV that varies from one link to another link. To capture such randomness, MD of the CP is defined as \cite{haenggi2015meta}
\begin{align}
&\overline{F}_{\mathrm{P_{c}}(\tau)}(x)=\prob{\mathrm{P_{c}}(\tau)>x}, 
\end{align}where $\mathrm{P_{c}}(\tau)$ is defined as
\[\mathrm{P_{c}}(\tau)=\prob{\SINR >\tau\vert\Phi_{\rm a}}.\]
To compute the MD, we first determine the $q$th moment of the CP for a given realization $\Phi_{\rm a}$. Subsequently, using the Gil-Pelaez lemma \cite{Gil1951note}, we derive the MD of the CP. The following lemma provides the $q$th moment of ${\rm P_{c}}(\tau)$ conditioned the distance $R$ of the serving RSU.
\begin{lemma}\label{lemma4}
	The $q$th moment $M_{q}(\tau)$ of the CP is
	\begin{align}
		&M_{q}(\tau)=\int_{r=0}^{\infty}\expS{-\frac{2\prm_{\cdot}\lambda_{\rsu} r}{\alpha}\int_{0}^{1}\left(1-\frac{1}{\left(1+\tau y\right)^{q}}\right)y^{-\eta}\dv y}\nonumber\\
		&\quad\times e^{-\frac{q\tau r^{\alpha}}{\SNR}}f_{R}(r)\dv r, \label{MGF1}
	\end{align}where  $\eta={(1+\alpha)}/{\alpha}$.
\end{lemma}
\begin{proof}
The $q$th moment of $P_{\rm s}(\tau)$ conditioned on $R$ can be written as
		\begin{align*}
	&M_{q}(\tau\vert R)=\expects{\Phi_{\rm a}\vert R}{\prod\nolimits_{\z_{i}\in \Phi_{\rm a}}\frac{1}{\left(1+\tau R^{\alpha}|\z_{i}|^{-\alpha}\right)^{q}}e^{-\frac{q\tau R^{\alpha}}{\SNR}}}\\
	&=\expS{-2\lambda_{\rsu}\prm_{\cdot}\int_{z=r}^{\infty}\left(1-\frac{1}{\left(1+\tau r^{\alpha}z^{-\alpha}\right)^{q}}\right)\dv z}e^{-\frac{q\tau r^{\alpha}}{\SNR}}.
\end{align*}
To simplify further, let us replace $z^{-\alpha}= y r^{-\alpha}\implies\dv z=-\frac{1}{\alpha}ry^{\frac{-1-\alpha}{\alpha}}\dv y$. Deconditioning with respect to $R$ completes the proof of Lemma.
\end{proof}Using the Gill-Paleaz lemma \cite{Gil1951note} for  inversion of the MGF, we derive the MD $\overline{F}_{\mathrm{P_{c}}(\tau)}(x)$ of the CP which is given as
\begin{align}\label{MD proof}
	&\overline{F}_{\mathrm{P_{c}}(\tau)}(x)=\frac{1}{2}+\frac{1}{\pi}\int_{0}^{\infty}\frac{\mathrm{Im}\left(e^{-\irm t \ln (x)} M_{\irm t}(\tau)\right)}{t}\dv t.
\end{align} In the following theorem, we  present the MD of CP. 

\begin{theorem}	
	The MD $	\overline{F}_{\mathrm{P_{c}}(\tau)}(x)$ of the CP is (for proof, see Appendix \ref{proof-meta-cov})
	\begin{align}
	&=\frac{1}{2}-\frac{1}{\pi}\int_{t=0}^{\infty}\frac{\int_{r=0}^{\infty}f_{\mathrm{r}}(t)\sin\left(f_{\irm}(t)+\Theta(t)\right)f_{R}(r)\dv r}{t}\dv t,\nonumber\\
	&\text{where }f_{\mathrm{r}}(t)=\exp\left(-({2\prm_{\cdot}\lambda_{\rsu} r}/{\alpha})\int_{0}^{1}\right.\nonumber\\
	&\left.{\left(1-\cos\left(t\ln\left(1+\tau y\right)\right)\right)}{y^{-\eta}}\dv y\right),\label{fr}\\ &f_{\irm}(t)=({2\prm_{\cdot}\lambda_{\rsu}r}/{\alpha})\int_{0}^{1}{\sin\left(t\ln\left(1+\tau y\right)\right)}{y^{-\eta}}\dv y, \label{fi}\\
	&\Theta(t)={t\tau r^{\alpha}}/{\SNR}+t\ln(x)\, \text{and}\,\eta=1+{1}/{\alpha}.\label{the}
\end{align} 
\end{theorem}
\textbf{RC analysis:}
{The RC for the typical VU is defined as the probability that the rate at the typical VU is greater than a certain threshold $\tau$, \ie
\begin{align}\label{defintionratecoverage}
	\mathrm{r}_{\rm c}(\tau)=\prob{\mathcal{R}_{\rm c}>\tau},
\end{align}where the rate  $\mathcal{R}_{\rm c}=\frac{B}{\widetilde{S}_{\cdot}+1}\log(1+\SINR)$, with $B$ denoting the available bandwidth at the RSU and $\widetilde{S}_{\cdot}$ denoting the load on the tagged cell. Recall that, $\widetilde{S}_{\pt}$ and $\widetilde{S}_{\Nrm}$ denote the cell load in PTS and N-PTS, respectively. The RSU load directly affects the RC at the typical VU. Using the load distribution of PTS and N-PTS, we derive the RC and its MD. 
\begin{theorem}
	The RC at the typical VU in PTS and N-PTS is
	\begin{align*}
		\mathrm{r_{c}}(\tau)=\sum\nolimits_{k=0}^{\infty}\prob{\widetilde{S}_{\cdot}=k}\mathrm{p_{c}}(\gamma({\tau (k+1)}/{B})),
	\end{align*}where $\gamma(x)=2^{x}-1$.
\end{theorem}
\begin{proof}
	From the definition of the RC, we know that
		\begin{align*}
		\prob{\mathcal{R}_{\rm c}>\tau\vert \widetilde{S}_{\cdot}}&=\prob{\left(\frac{B}{\widetilde{S}_{\cdot}+1}\log(1+\SINR)>\tau\right)\vert\widetilde{S}_{\cdot}}\\
		&=\mathbb{P}\left(\SINR>\gamma\left({\tau (\widetilde{S}_{\cdot}+1)}/{B}\right)\right).
	\end{align*}Finally, deconditioning using the distribution of $\widetilde{S}_{\cdot}$, we get the desired result.
\end{proof}
\begin{theorem}
The MD of RC is
\begin{align*}
	\overline{F}_{\mathrm{R_{c}}(\tau)}(x)=\sum\nolimits_{k=0}^{\infty}\prob{\widetilde{S}_{\cdot}=k}\overline{F}_{\mathrm{P_{c}}(\gamma({\tau (k+1)}/{B}))}(x),
\end{align*}where $\overline{F}_{\mathrm{P_{c}}(\gamma({\tau (k+1)}/{B}))}(x)$ is provided in Theorem 10. \end{theorem}
\begin{proof}
To derive the MD of RC, we need $q$th moment of RC which is defined as
\begin{align*}
	\mathrm{S}_{q}(\theta(\tau))&=\expect{\left(\prob{\mathcal{R}_{\rm c}>\tau}\vert \Phi_{\rm a},\widetilde{S}_{\cdot}\right)^{q}}\\
&	=\prob{M_{q}(\gamma({\tau (k+1)}/{B}))\vert \widetilde{S}_{\cdot}}.
\end{align*}Hence, MD for the RC $\overline{F}_{\mathrm{R_{c}}(\tau)}(x)$ is
	\begin{align*}
&=\frac{1}{2}+\frac{1}{\pi}\sum\nolimits_{k=0}^{\infty}\prob{\widetilde{S}_{\cdot}=k}\int_{0}^{\infty}\frac{\mathrm{Im}\left(e^{\irm t \ln (x)} M_{\irm t}(\theta(\tau))\right)}{t}\dv t.
	\end{align*}
\end{proof}
 \begin{figure}[t!]
	\centering
	\includegraphics[width=.8\linewidth]{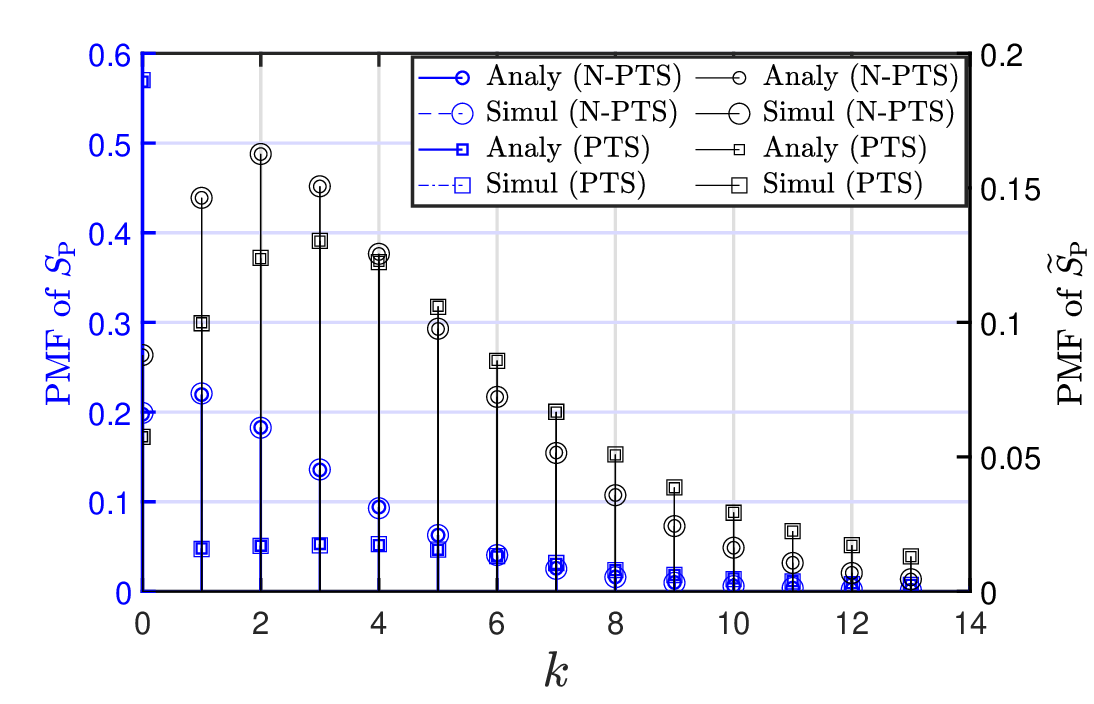}
	\caption{Verification of analytical results of load distribution on the typical and tagged RSU with simulation. The RSU density $\lambda_{\rsu}=2$ \si{RSU/km}, $a=100$ \si{m}, $\lambdaPT=1$ \si{platoon/km} and $m=5$ and for N-PTS $\lambda=5$ \si{VUs/km}.}
	\label{fig:analysimulloadtypicaltaggeditsptsverification}
\end{figure}
{\section{Numerical Results}
In this section, we present the numerical results describing the load distribution and its statistical properties, including mean, variance, and skewness. We then present the results for V2V communication, CP, RC, and their MDs for V2I communication. 
\subsection{Verification of analytical results}
In  Fig \ref{fig:analysimulloadtypicaltaggeditsptsverification}, we plot the load distribution on the typical and
tagged RSU by evaluating corresponding expressions along with numerical simulations. We observe that the simulation and analytical results of the load distribution for the typical and the tagged RSU in PTS and N-PTS match well.\\ For the typical RSU, the zero-load probability ($p_{S_{\cdot}}(0)$) is higher in PTS than in N-PTS. This occurs because platoon movement keeps VUs closely placed, leaving many RSUs idle. The load distribution in PTS is also more spread due to this uneven distribution implying higher variance of the distribution. For the load distribution on the tagged RSU, the probability of serving a single VU, equivalent to $p_{\widetilde{S}_{\cdot}}(0)$ as shown in Fig. \ref{fig:analysimulloadtypicaltaggeditsptsverification}, is lower in PTS than in N-PTS. This lower probability is due to a positive correlation between VU locations in PTS. 
\begin{figure}[t!]
		\centering
	(a)	\includegraphics[width=.8\linewidth]{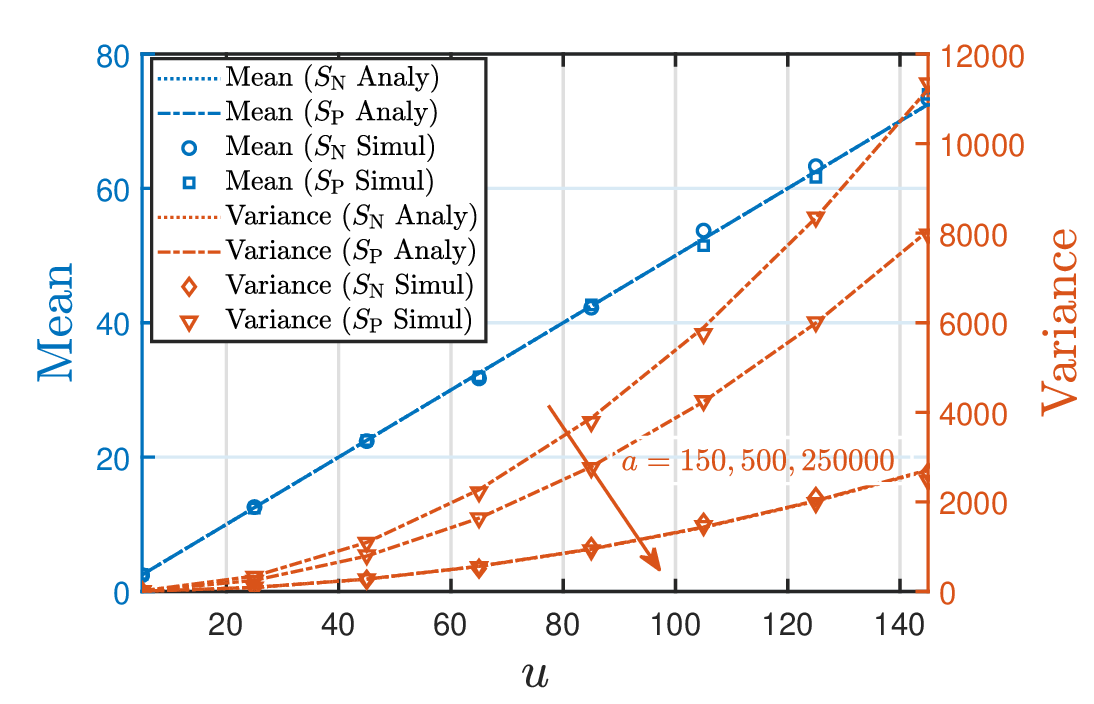}
	(b){\includegraphics[width=.8\linewidth]{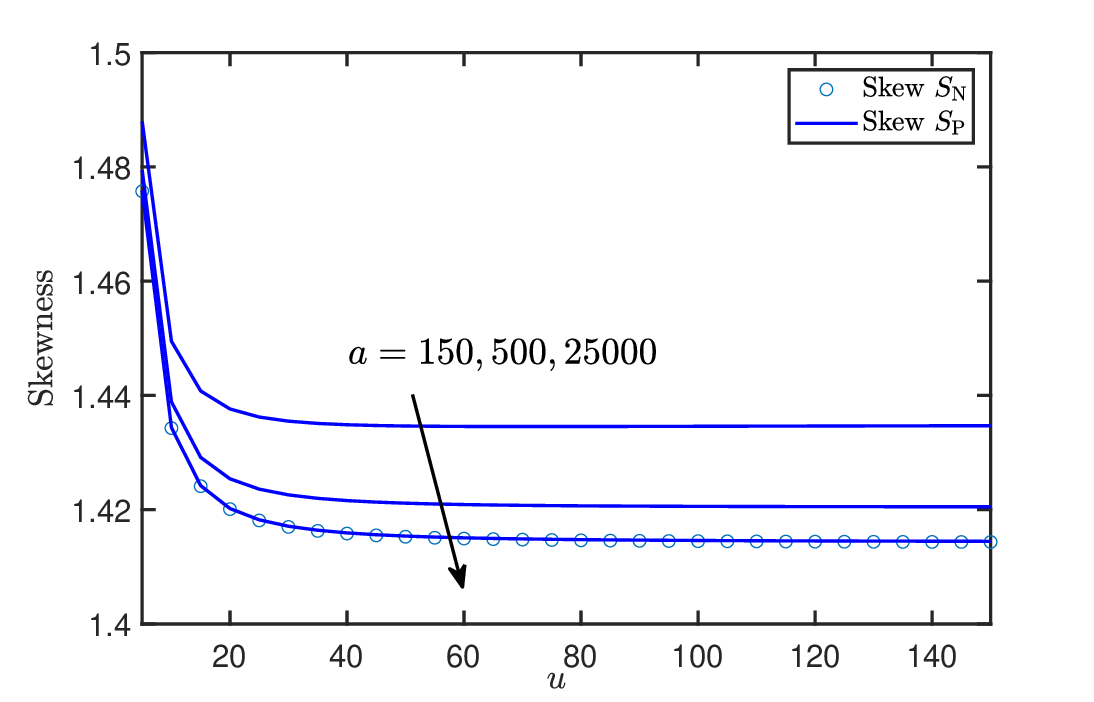} }
	\caption{The plot showing the analytical and simulation values of (a) mean and variance and (b) skewness for the load on the typical cell. The parameters are  $\lambda_{\rsu}=2$ \si{RSU/km}, $\lambda_{\pt}=1$ \si{platoon/km} for different values of $a$ in \si{meters}. }\label{fig:meanvariancetypicaln4}
\end{figure}
\subsection{Mean, variance and skewness for the load on the typical and the tagged RSU}\label{prev-sec}
In Figs. \ref{fig:meanvariancetypicaln4} and \ref{fig:meanvariancetaggedn3}, we have shown the analytical results and respective simulations for 
mean, variance and skewness of the load on the typical and tagged RSU. Here, we vary $u=m=\lambda/\lambdaPT$  while keeping $\lambdaPT$ fixed, which ensures that the overall VU density remains the same in both N-PTS and PTS. \\ As shown in Fig.~\ref{fig:meanvariancetypicaln4}(a), the mean load on the typical RSU is same for both PTS and N-PTS. However, the variance (Fig.~\ref{fig:meanvariancetypicaln4}(a)) and skewness (Fig.~\ref{fig:meanvariancetypicaln4}(b)) are higher in PTS. Since both distributions are positively skewed, most load values fall below the mean. Yet PTS, with its greater variance and skewness, produces a wider spread and a longer right tail. As a result, extreme load values—both high and low—occur more often in PTS. In contrast, N-PTS has a distribution more tightly clustered around the mean. Thus, RSUs under PTS are more likely to experience fluctuating loads.\\ As shown in Fig.~\ref{fig:meanvariancetaggedn3}(a), the tagged RSU in PTS has a higher mean load than in N-PTS, owing to the additional load from the tagged platoon. The variance and skewness of the load distribution are also greater in PTS. Because both distributions are positively skewed, the median load falls below the mean in both scenarios. As with the typical RSU, most of the probability mass lies below the mean. The higher variance in PTS reveals a more uneven load distribution, suggesting that the tagged RSU may serve more VUs as compare to N-PTS.

\begin{figure}[t!]
(a)\includegraphics[width=.8\linewidth]{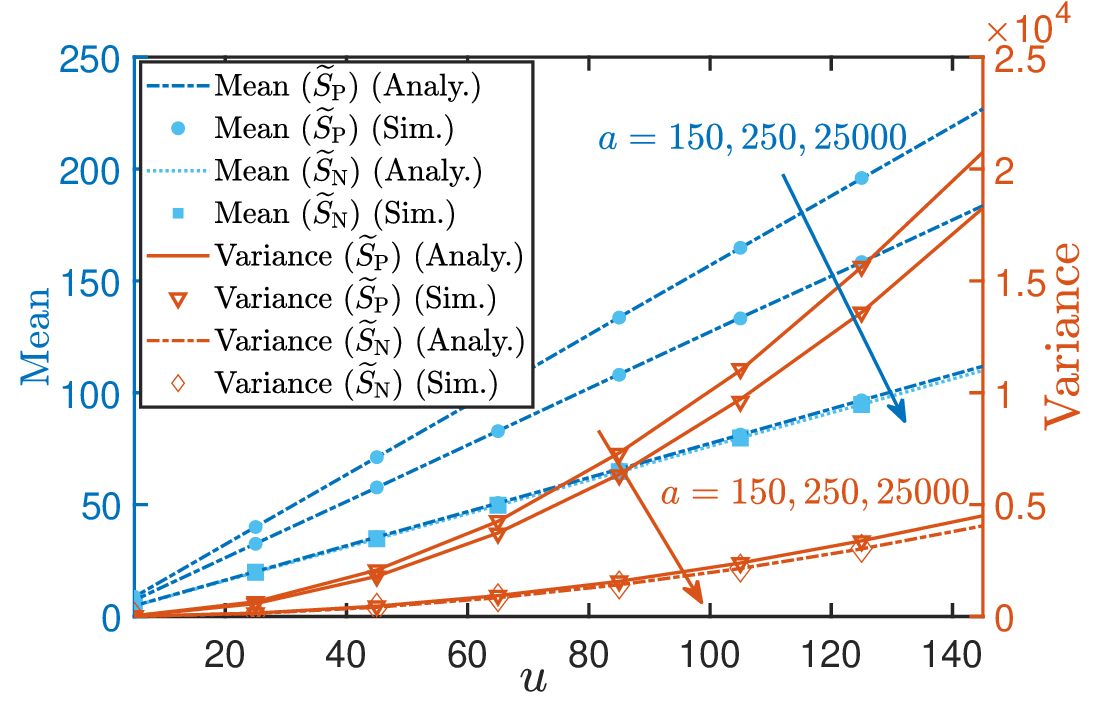}
(b){\includegraphics[width=.8\linewidth]{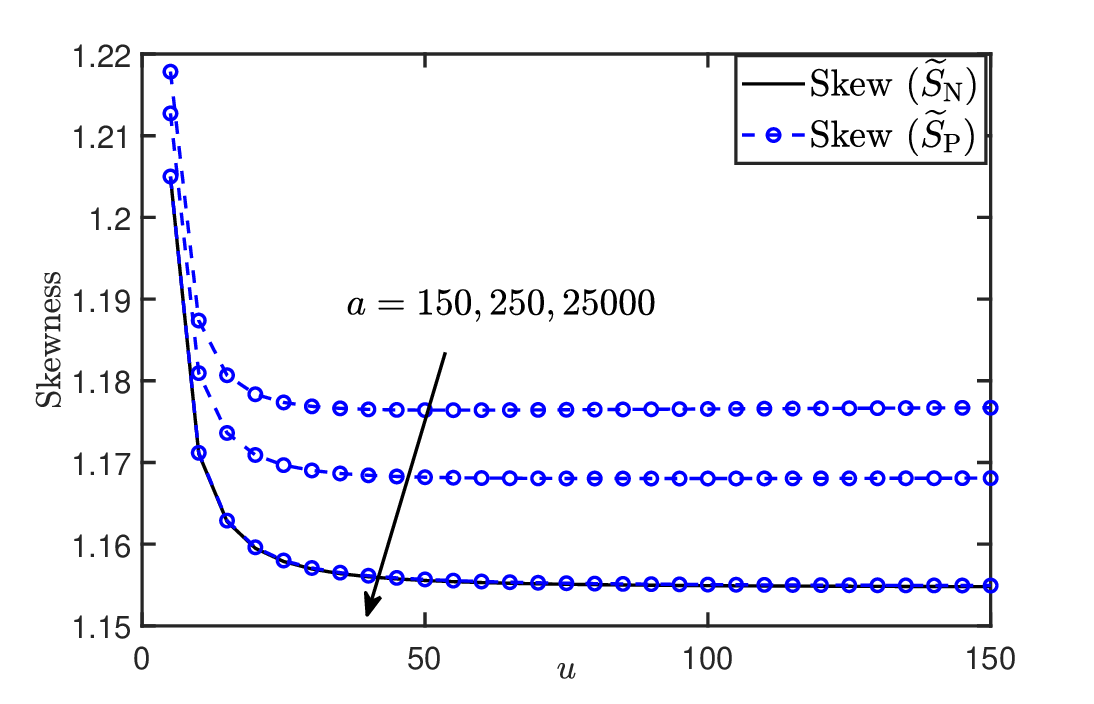} }
	\caption{The diagram showing the (a) mean and variance (b) skewness for the tagged cell. The parameters are $\lambda_{\pt}=1$ \si{platoon/km}, $\lambda_{\rsu}=2$ \si{RSU/km} and $a$ in \si{meters}. The mean and variance of load on the tagged cell is higher in PTS and converges to N-PTS.} 
	\label{fig:meanvariancetaggedn3}
\end{figure}
\subsection{Impact of platooning in the urban residential and rural freeways}
As per vehicular speed recommendations in various states in US (e.g. see traffic data for the state of Tennessee \cite{tensee}), traffic can be classified into rural traffic environment (RTE) and urban traffic environment (UTE). In RTE, the recommended average speed for vehicles is $96\mbox{--}112$ \si{km/hr} ($60\mbox{--}70$ \si{mi/hr}). Considering an average distance of $80\mbox{--}100$ \si{m} between VU, including stopping distance \cite{shi2020effects}, and vehicle length, the average number of VUs in RTE is $10\mbox{--}15$ \si{VU/km}. Similarly, for UTE, the average speed is $48\mbox{--}64$ \si{km/hr} ($30\mbox{--}40$ \si{mi/hr}). Following similar calculation as performed in RTE, the average number of VUs in UTE is $25\mbox{--}35$ \si{VUs/km}. To compare N-PTS with PTS in different traffic conditions, we vary the VU density of N-PTS $\lambda$ from RTE to UTE with maintaining $u=m=\lambda/\lambdaPT$ keeping $\lambdaPT=1$ \si{platoon/km}. With a similar VU density, analyzing driving patterns provides insights into the effects of switching between N-PTS and PTS in UTE or RTE. Next, we examine the load distribution for UTE and RTE on the typical and tagged RSUs.}
\begin{figure}[t!]
	\centering
	(a){\includegraphics[width=.8\linewidth]{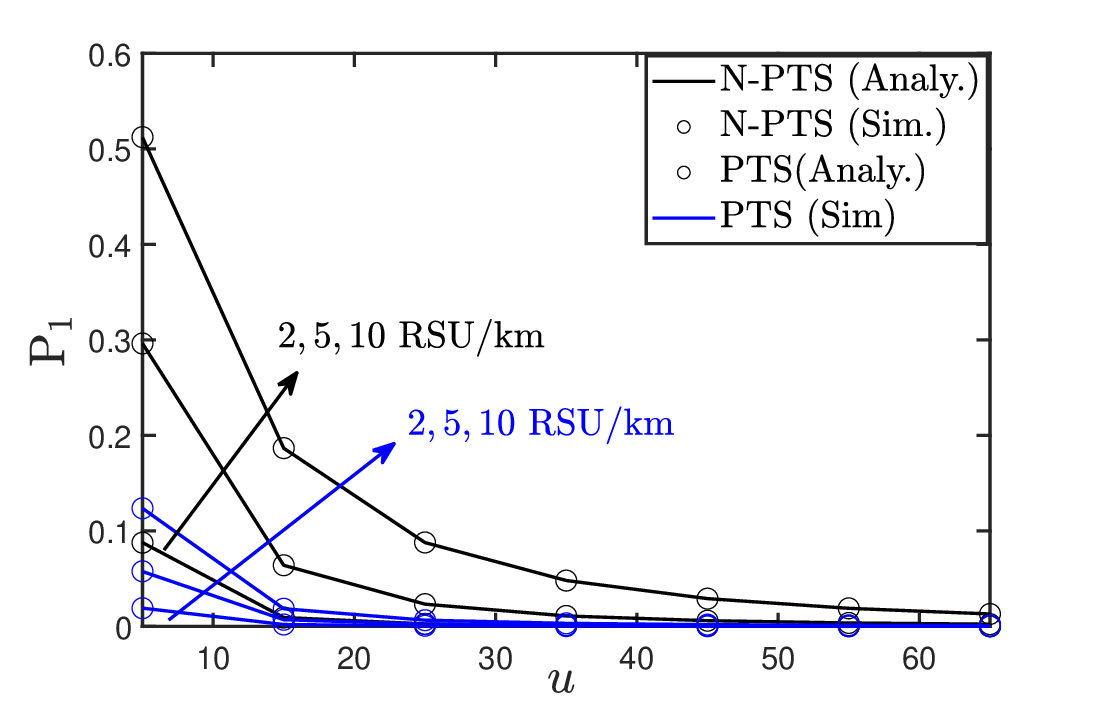} }
	(b){\includegraphics[width=.8\linewidth]{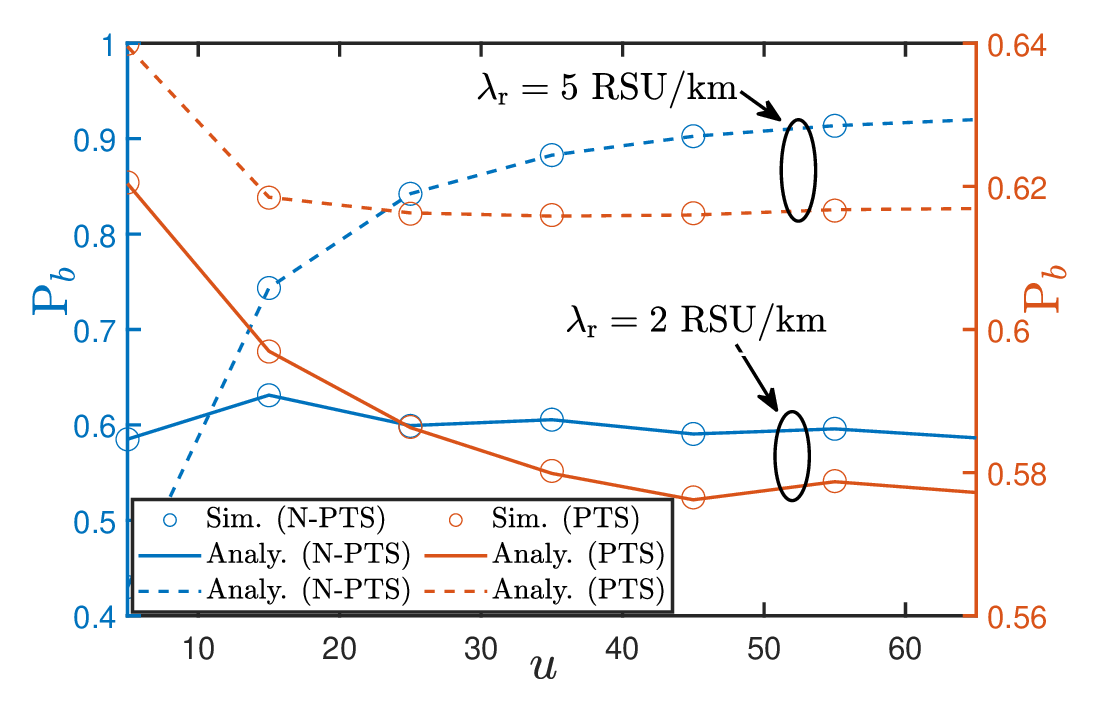} }
	\caption{The plot showing (a) the probability that an RSU is serving a single VU and (b) tagged-below-average-loading probability $\mathrm{P}_{b}$. The parameters are $a=100$ \si{m}, $\lambdaPT$=1 \si{RSU/km}.}
	\label{fig:onprobtagged}
\end{figure}
{\subsubsection{Load  balance on the typical cell in RTE and UTE}
To understand the nature of the distribution of the load on the typical RSU under the two different traffic conditions, we first define the RSU's off probability $\prm_{\mathrm{off}}$ and below-average-loading probability $p_{b}$. From Theorem \ref{thm:8}, the off probability is $\prm_{\mathrm{off}}=1-\prm_{\cdot}$ Before defining the below-average-loading probability, let\[s_{avg}=\mathbb{E}\left[S_{\cdot}\vert S_{\cdot}>0\right]={\mathbb{E}\left[S_{\cdot}\right]}/{\prm_{\cdot}}={\mathbb{E}\left[S_{\cdot}\right]}/{(1-\prm_{\mathrm{off}})},\]which denotes the mean load on the typical RSU conditioned that the typical RSU is serving at least a vehicle. Let us round down the $s_{avg}$ as $k_{\avg}=\Big\lfloor	\mathbb{E}\left[S_{\cdot}\vert S_{\cdot}>0\right]\Big\rfloor.$ Now, we can define the below-average-loading probability as the probability that the typical RSU is serving less than or equal to the average load $k_{\avg}$ conditioned that the typical RSU is serving at least a vehicle \ie $p_{b}=\sum_{k=1}^{k_{\avg}}p_{S_{\cdot}}(k)$. The below-average-loading probability is determined based on the principle that the network is designed to ensure a sufficient per-user rate when the number of VUs in the typical RSU’s serving region is below the mean load. Hence, a higher value of $p_{b}$ denotes that most of the RSUs are working under the safe operational conditions. Now, we first analyze the off probability of the typical RSU followed by the analysis of below-average-loading probability.
In Fig. \ref{5a}, we have shown $\mathrm{p}_\mathrm{off}$ with varying VU density from RTE (lower VU density) to UTE (higher VU density). We observe that the off probability for the typical RSU in PTS is higher as compared to N-PTS. In case of N-PTS, there is a significant drop in $\mathrm{p_{off}}$ from RTE to UTE. For N-PTS, as $\mathrm{p_{off}}$ decreases from RTE to UTE, more RSUs become active. This increased number of active RSUs increases the interference for the typical VU, ultimately reducing the network's overall coverage. In PTS, the value of $\mathrm{p_{off}}$ varies slightly with different traffic conditions (from RTE to UTE), resulting in less impact of interference at the typical VU.\\
In Fig. \ref{fig:offprobjoint}, we plot the below-average-loading probability $p_{b}$ with $u$. 
The below-average-loading probability $p_{b}$ increases with $u$ in N-PTS compared to PTS. In PTS, small variation of $p_{b}$ across VU density indicates that an RSU handling below-average-loading remains steady from RTE to UTE, unlike N-PTS, which changes with VU density. Also the mean load on the typical RSU can be reduced by increasing the RSU density. However, increasing RSU density in PTS does not increase interference at the typical VUs due to the higher $\mathrm{p_{off}}$ compared to N-PTS, keeping most RSUs silent.}
 \begin{figure}[t!]
	\centering
	\includegraphics[width=.8\linewidth]{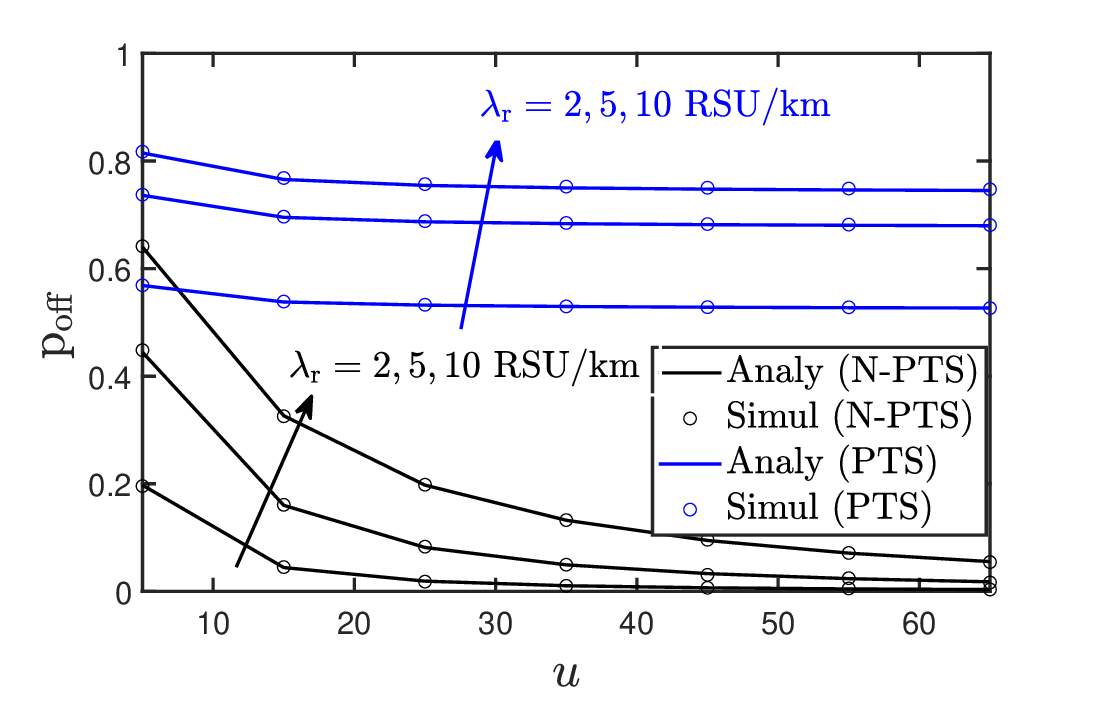}
	\caption{The off probability for N-PTS and PTS 
		with respect to the VU density ranging from the RTE to UTE. The platoon radius $a=100$ \si{m} with $u=m=\lambda/\lambdaPT$ and  $\lambdaPT=1$ \si{platoon/km}.}\label{5a}
	{\includegraphics[width=.8\linewidth]{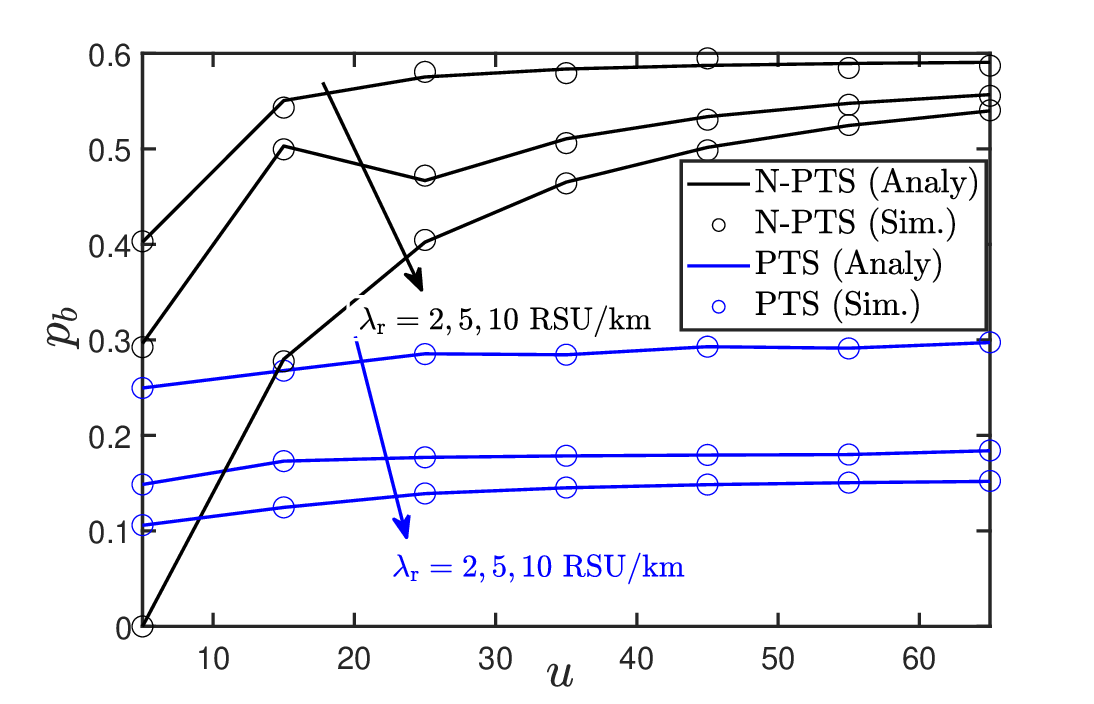} }
	\caption{ The plot showing the $p_{b}$ with the $u$. The parameters are $\lambdaPT=1$ \si{platoons/km} and $a=100$ \si{m}.}
	\label{fig:offprobjoint}
\end{figure}
\newcommand{\Pmin}{\mathrm{P}_{1}}
\newcommand{\pavg}{\mathrm{P}_{b}}
\subsubsection{Load balance on the tagged cell in the UTE and RTE}
In PTS, since VUs are closely located, the tagged RSU is unlikely to serve only a single vehicle. Furthermore, it is recommended that a serving RSU should not use all its resources to serve a VU. To understand the effect of load distribution on the tagged RSU, we define minimum load probability $\Pmin$, $m_{\mathrm{avg}}$ and tagged-below-average-loading probability $\mathrm{P}_{b}$. Mathematically, we can write them as
\[
\mathrm{P}_{1}=p_{\widetilde{S}_{\cdot}}(1),\,
m_{\mathrm{avg}}=\Bigl\lfloor\mathbb{E}\left[\widetilde{S}_{\cdot}\right]\Bigr\rfloor,\,\mathrm{P}_{b}=\sum\nolimits_{k=1}^{m_{\mathrm{avg}}}p_{\widetilde{S}_{\cdot}}(k),\] recall that, $\mathbb{E}[\widetilde{S}_{\cdot}]$ denotes the average load on the tagged RSU.

As shown in Fig. \ref{fig:onprobtagged}(a), the minimum load probability $\mathrm{P}_{1}$ is higher in N-PTS compared to PTS. Therefore in PTS,  an operational RSU supports a higher number of VUs, which also facilitates platoon coordination and communication. From a traffic standpoint, $\mathrm{P}_{1}$ remains low from RTE to UTE in PTS. In contrast to PTS, in N-PTS the $\mathrm{P}_{1}$ fluctuates substantially from RTE to UTE. It is obvious that as the VUs on the road increases it is unlikely that a tagged RSU serves a single vehicle in any traffic scenario (PTS or N-PTS). In Fig. \ref{fig:onprobtagged}(b), we observe that higher value of $\mathrm{P}_{b}$ in N-PTS compared to PTS indicates that a serving RSU is more likely to handle a load lower than the average load $m{\avg}$ in N-PTS.
Moreover, increasing RSU density raises $\mathrm{P}b$, indicating that in a PTS, each RSU must be able to handle a load exceeding $m{\text{avg}}$ 
\subsubsection{Analysis of connectivity degree} 
{We now analyze which traffic pattern (PTS or N-PTS) provides better connectivity for the typical VU across different communication range $R_{\rm b}$ values and range of VU densities $u$. We define connectivity degree as the number of VUs within the typical VU's communication range. Further, we define the probability $p_{\mathrm{s}}$
as the probability that connectivity degree is more than $k$. Hence, $p_{\mathrm{s}}=\mathbb{P}[N_{\cdot}>k]=1-\mathbb{P}[N_{\cdot}\leq k]$.
Fig. \ref{fig:nodedegreeplot} shows that in PTS, a vehicle is better connected to nearby VUs due to a higher $p_{\mathrm{s}}$ as compared to N-PTS. Consequently, in the event of an emergency, platooning
delivers better connectivity as compared N-PTS.
\begin{figure}[t!]
	\centering
	\includegraphics[width=.8\linewidth]{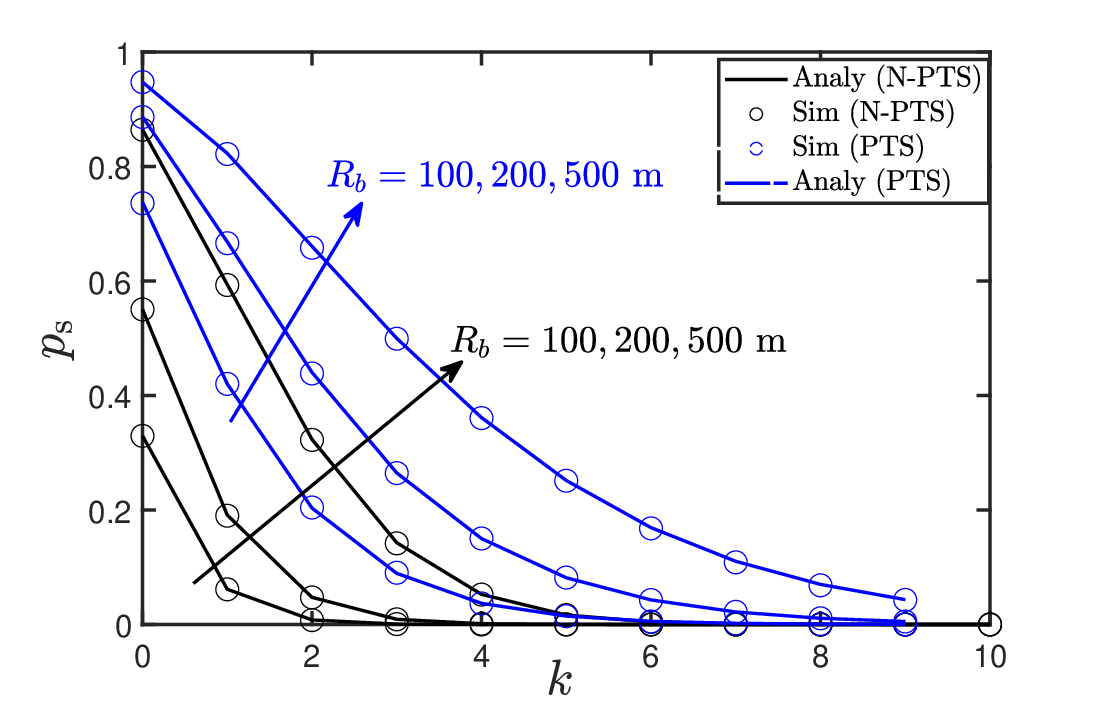}
	\caption{Variation of $p_\mathrm{s}$ with $k$. It is clear that VUs are well connected with near by VUs in PTS. The parameters are $a=150$ \si{m}, $\lambdaPT=1$ \si{platoons/km} and $u=2$.}
	\label{fig:nodedegreeplot}
\end{figure}
}
\subsection{CP and MD of coverage probability} 
Fig. \ref{fig:covprobanalysimulonprob}(a) shows the variation of the CP for PTS and N-PTS with the VU density. The CP for the typical VU in PTS is higher than in N-PTS. The higher active probability in N-PTS (see Fig. \ref{fig:covprobanalysimulonprob}(a)) results in more RSUs transmitting signals causing increased interference to the typical VU in N-PTS.  To analyze the link-level impact, Fig. \ref{fig:covprobanalysimulonprob}(b) presents the MD of the CP. The MD also highlights the disparity in CP between PTS and N-PTS, confirming that PTS provides better coverage for VUs.
 \subsection{RC and MD of RC}
In the previous subsection, we noted that the CP of the typical VU is higher in PTS than in N-PTS. At the same time, load distribution analysis showed that the typical RSU in PTS carries a higher mean load. As shown in Fig.~\ref{fig:metadistrratecoverageptsnpts}, the RC for the typical VU is lower in PTS. This outcome is expected, as heavier RSU loads in PTS reduce their capacity to support higher rates. To examine effects at the individual link level, we also plot the MD of the RC. Interestingly, the MD differs little between the two scenarios. For certain VU densities, the MDs are nearly identical. This can be explained by the uneven VU distribution in PTS, where reduced interference in less congested regions allows some users to achieve better RC than in N-PTS.
	 
\begin{figure}[t!]
	 	\centering
(a)\includegraphics[width=.8\linewidth]{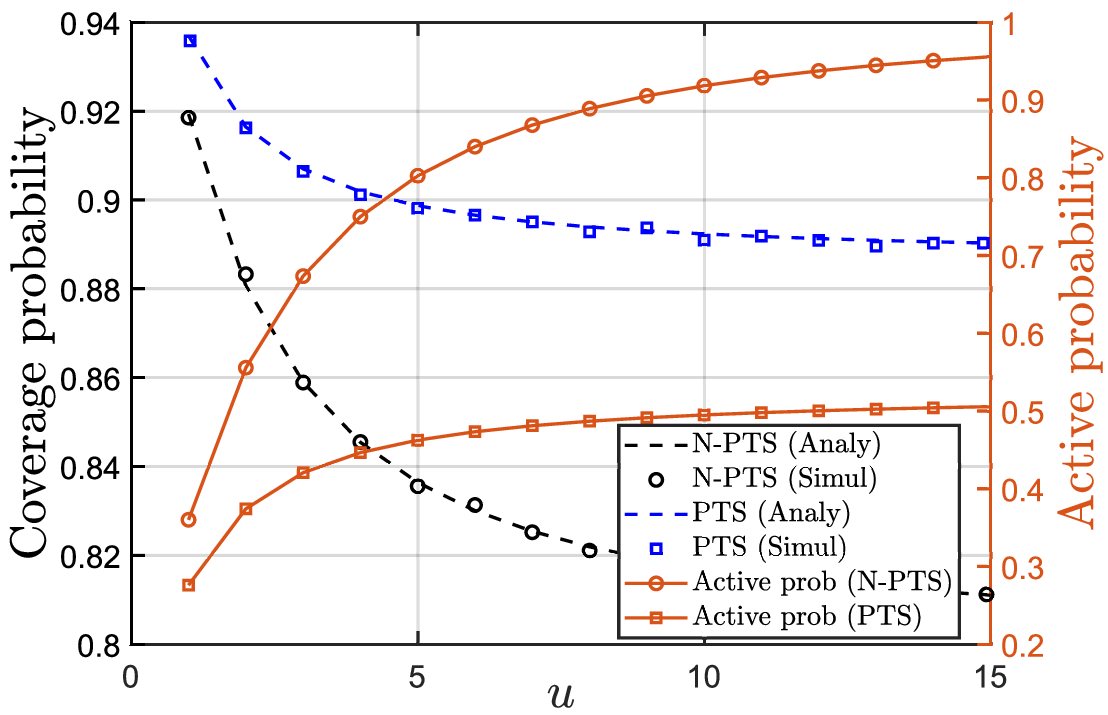}
(b)\includegraphics[width=.8\linewidth]{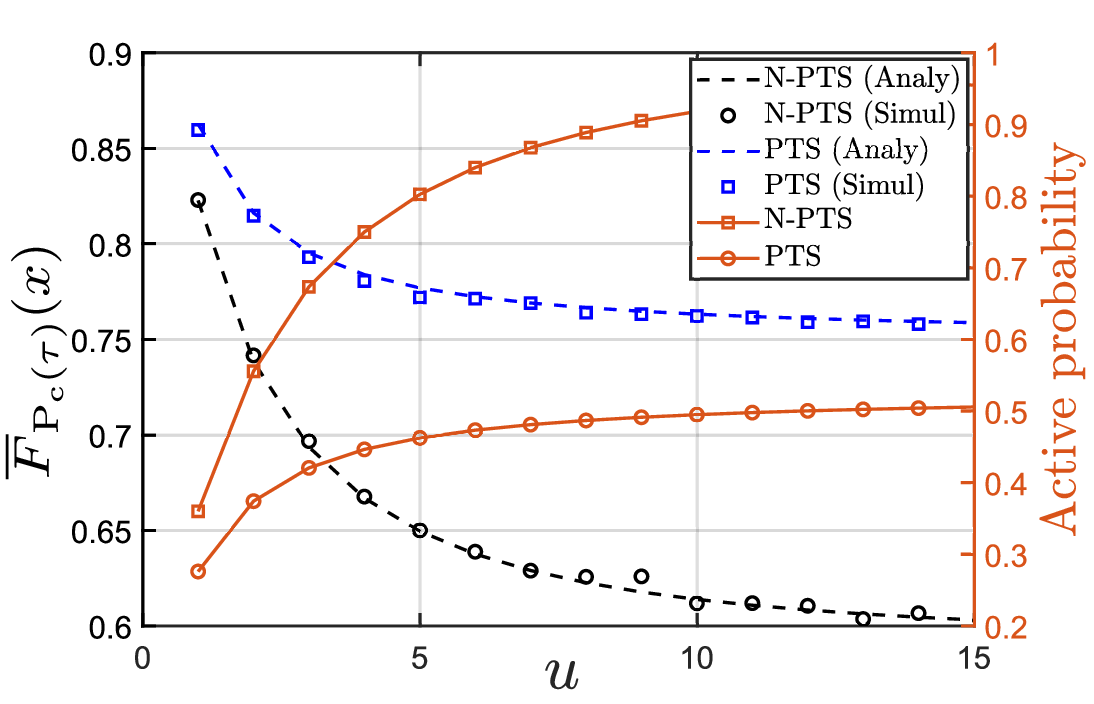}
	 	\caption{Variation of (a) CP and active probability with  $u=\lambda/\lambdaPT$ and $m=u$. (b) MD of CP. The parameter values are $\lambda_{\rsu}=2$ \si{RSU/km}, $\lambda_{\pt}=1\, \si{platoon/km}, a=150\, \si{m}$ , $P_{t}=1$ \si{W}, $\sigma^{2}=0.00005$ \si{W}, SINR threshold $\tau=0.9$, MD threshold $x=0.8$, and pathloss coefficient $\alpha=3.5$.}\label{fig:covprobanalysimulonprob}
	 	\centering
	 	\includegraphics[width=.8\linewidth]{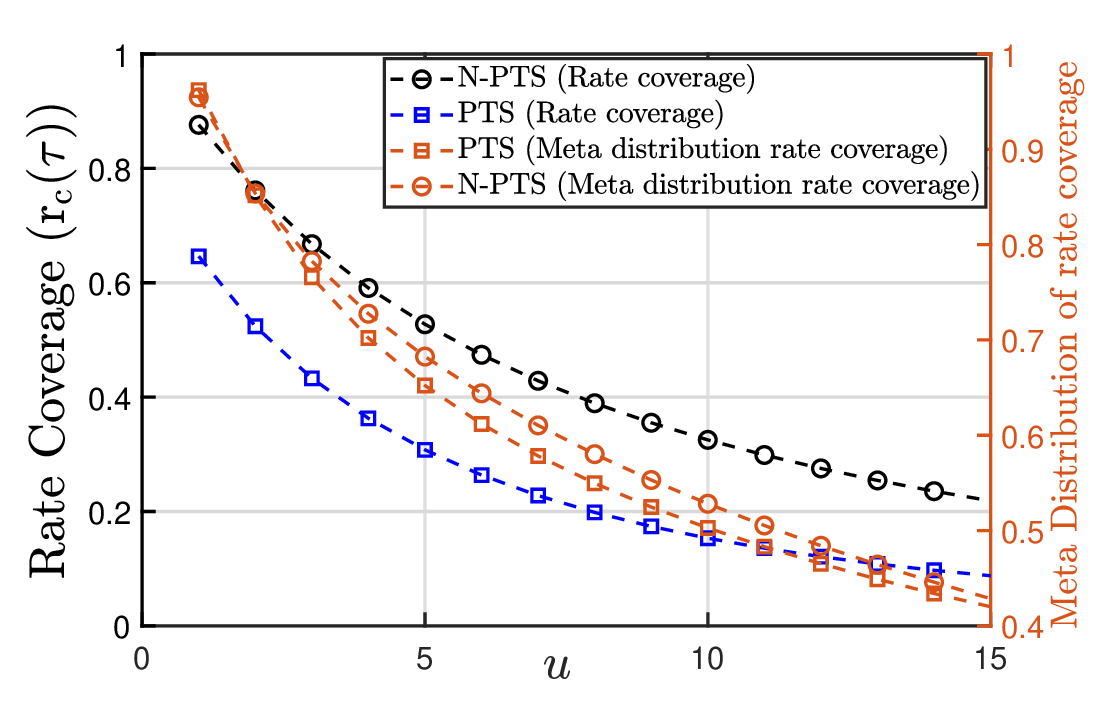}
	 	\caption{Variation of MD of RC with the platoon density $\lambda/\lambda_{\pt}=u$ and $m=u$. The parameter values are $\lambda_{\pt}=1$ \si{platoon/km}, $a=150$ \si{m}, $\lambda_{\rsu}=2$ \si{RSU/km}, transmit power $P_{t}=1$ \si{W}, $\sigma^{2}=0.00005$ \si{W}, bandwidth to available to each RSU=10 \si{MHz}, rate threshold $\tau=9$ \si{MHz}, $\alpha=4$, and rate threshold for MD $x=0.9$.}
	 	\label{fig:metadistrratecoverageptsnpts}
	 \end{figure}	 
\section{Conclusion and Future Work}
This paper studied the impact of platooning on two forms of vehicular communication—RSU-based V2I and direct V2V—in a highway network. VU locations were modeled using the MCP for PTS and a PPP for N-PTS, while RSUs were modeled as an independent 1D PPP along the road. Within this framework, we derived closed-form expressions for the PGF and PMF of the load on both typical and tagged RSUs, along with their mean, variance, and skewness. Numerical results show that in PTS, a higher off probability leaves most RSUs inactive, while a lower below-average probability ensures that active RSUs carry loads above the mean, improving their efficiency. For V2V connectivity, we derived the PGF and PMF of the connectivity degree under both traffic scenarios. We also analyzed the CP and RC for PTS and N-PTS. PTS achieves better CP due to reduced interference from inactive RSUs, but suffers from lower RC because active RSUs handle heavier loads. To assess link-level reliability, we examined the MD of the RC. For some VU densities, the MD is nearly the same for both scenarios. This similarity stems from the uneven distribution of VUs in PTS, which allows some users to experience better RC than in N-PTS.\\Future work could explore RSU deployment depending
on the traffic, where intensity varies with the traffic condition.
Extending the analysis to urban environments or multi-lane highways
can provide better insights into optimizing connectivity.
Investigating heterogeneous data traffic requirements, such
as latency-sensitive and high-bandwidth applications, could enhance
network efficiency.

\appendices
\section{}\label{proof_of_lemma1}
To solve the integrals in \eqref{i1nk} and \eqref{i2nk}, we first determine $\kappa^{n}(r/2,k)$ from \eqref{kappadefin2}. Depending on the range of $r$, $\kappa^{n}(r/2,k)$ is given as
	\begin{align}\label{kappn}
		&\!\!\!\!\kappa^{n}(r/2,k)=
		\begin{cases}
			(2\lambdaPT)^{n}(\lambda_{\drm})^{nk}a^{n}\times\\
			\sum_{j=0}^{n}{n \choose j}\left(\frac{ (1-k)}{2a(1+k)}\right)^{j}r^{nk+j}, & r<2a,\\
			\left(2\lambdaPT m^{k}\right)^{n}\eta^{n}_{1}\sum_{j=0}^{n}\frac{{n \choose j}}{(2\eta_{1})^{j}} r^{j},& r>2a,
		\end{cases}
	\end{align}where $\eta_{1}=(a(1-k))/(1+k)$. We know that $f_{L_\ob}(r)=\lambda_{\rsu} r f_{L}(r)$. Replacing \eqref{kappn} in \eqref{i1nk}, we get
	\begin{align}
	&I(n,k)=\int_{0}^{2a}\hspace{-.5em}	\kappa^{n}(r/2,k)f_{L}(r)\dv r+		\int_{2a}^{\infty}\nonumber	\kappa^{n}(r/2,k)f_{L}(r)\dv r\\
	&\stackrel{(a)}=I_{1}(n,k)+I_{2}(n,k).\label{Ink}
	\end{align}Now, we individually solve $I_{1}(n,k)$ and $I_{2}(n,k)$ presented in step $(a)$. The integral $I_{1}(n,k)$ can be written as
	\begin{align*}
		&I_{1}(n,k)=(2\lambdaPT)^{n}(\lambda_{\drm})^{nk}a^{n}\sum\nolimits_{j=0}^{n}{n \choose j}\left(\frac{ (1-k)}{2a(1+k)}\right)^{j} 4 \lambda_{\rsu}^{2}\nonumber\\
		&\int_{0}^{2a}r^{nk+j+1}e^{-2\lambda_{\rsu}r}\dv r.
	\end{align*}Similarly, the second part $I_{2}(n,k)$ is
\begin{align*}
	&=\left(2\lambdaPT m^{k}\right)^{n}\eta^{n}_{1}\sum_{j=0}^{n}\frac{{n \choose j}}{(2\eta_{1})^{j}}(2\lambda_{\rsu})^{2}  \int_{2a}^{\infty}r^{j+1}e^{-2\lambda_{\rsu}r}\dv r.
\end{align*}Solving the integral  and replacing $I_{1}(n,k)$ and $I_{2}(n,k)$ in \eqref{Ink} completes the proof of \eqref{i1nk}. Similar to ${I}(n,k)$, $\widetilde{I}(n,k)$ can be written as
\begin{align*}
	\widetilde{I}(n,k)=\widetilde{I}_{1}(n,k)+\widetilde{I}_{2}(n,k).
\end{align*}Using the method adopted to derive ${I}_{1}(n,k),\, {I}_{2}(n,k)$, we can derive $\widetilde{I}_{1}(n,k),\,\widetilde{I}_{2}(n,k)$.

\section{}\label{proof_of_cor_2}
We know that the average size of the Voronoi cell is ${1}/{\lambda_{\rsu}}$. Further, the average number of points of $\Phi_{\mrm}$ in unit length is $m\lambdaPT$. Hence, the mean number of points in the Voronoi cell is $\mathbb{E}\left[S_{\pt}\right]={m\lambdaPT}/{\lambda_{\rsu}}$. Now for variance of $S_{\pt}$
	\begin{align*}
		\mathrm{Var}[S_{\pt}]&\stackrel{(a)}=\int_{t=0}^{\infty}\lim_{s\rightarrow1}\left((g^{(1)}(s,t/2))^{2}+g^{(2)}(s,t/2)\right)f_{L}(t)\dv t\nonumber\\
		&+{m\lambdaPT}/{\lambda_{\rsu}}-\left({m\lambdaPT}/{\lambda_{\rsu}}\right)^{2}\\
		&\stackrel{(b)}=\int_{t=0}^{\infty}\kappa^{2}(t/2,1)f_{L}(t)\dv t+\int_{t=0}^{\infty}\kappa(t/2,2)f_{L}(t)\dv t\nonumber\\
		&+{m\lambdaPT}/{\lambda_{\rsu}}-\left({m\lambdaPT}/{\lambda_{\rsu}}\right)^{2},
	\end{align*}where step $(a)$ is obtained from \eqref{variance}, step $(b)$ is obtained using \eqref{kappadefin}.  Furthermore, the two integrals in step $(b)$ are easily solved using the results in \eqref{i1nk}. Now, we present the third moment of $S_{\pt}$ which is given as 
\begin{align*}
	\mathbb{E}\left[S_{\pt}^{3}\right]=\left[\mathcal{P}^{(3)}_{S_{\pt}}(s)\right]_{s=1} + 3\left(\mathrm{Var}[S_{\pt}]+\left(\mathbb{E}[S_{\pt}]\right)^{2}\right) -2\mathbb{E}[S_{\pt}].
\end{align*}Obtaining the third derivative of $e^{g(s,r/2)}$ and applying the limit $s\rightarrow 1$, we get
	\begin{align*}
		\frac{\dv^{3} e^{g(s,l/2)}}{\dv s^{3}}&\stackrel{(a)}=\lim_{s\rightarrow1}\left[\left(g^{(1)}(s,l/2)\right)^{3}+g^{(3)}(s,l/2)\right.\\
		&\left.+3g^{(1)}(s,l/2)g^{(2)}(s,l/2)\right]\\
		&\stackrel{(b)}=\kappa^{3}(l/2,1)+\kappa(l/2,3)+3\kappa(l/2,1)\kappa(l/2,2)\\
		&\stackrel{(c)}=I(3,1)+I(1,3)+3I(1,1)I(1,2),
	\end{align*}where step $(a)$ is obtained by taking the third derivative of $e^{g(s,r/2)}$, step $(b)$ is obtained after applying the limit $s\rightarrow1$ then using \eqref{kappadefin} and finally step $(c)$ is obtained by deconditioning using the distribution of length $L$ and the results from \eqref{i1nk}. Replacing the variance and mean of $S_{\pt}$, we get the third moment of $S_{\pt}$.

\section{}\label{proofofpmf}
As we know that $x_{\road}\sim\mathcal{U}(0,t/2)$ and $x_{\ob}\sim\mathcal{U}(-a,a)$ which can be written as
\begin{align*}
	f_{x_{\road}}(x)=\frac{2}{t}\left(u(x)-u(x-t/2)\right)\\
	f_{	x_{\ob}}(x)={u(x+a)-u(x-a)}/{(2a)}.
\end{align*}Depending on the distance $|x_{\road}-x_{\ob}|$ between the two centers, the number of VUs is Poisson distributed with mean $m\bar{\A}(t/2,a,|x_{l}-x_{\ob}|)$. Hence, the PGF for $V_{\rm m}(t/2)$ is 
\begin{align*}		&\mathcal{P}_{V_{\mrm}(t/2)}(s)=\int_{-a}^{a}\int_{-t/2}^{t/2}\frac{1}{2at}{e^{m\bar{\A}(t/2,a,|x_{l}-x_{\ob}|)(s-1)}}{}\dv x_{l}\dv x_{\ob}\\
	&\stackrel{(a)}=\int_{-a}^{a}\int_{0}^{t/2}\frac{1}{at}{e^{m\bar{\A}(t/2,a,|x_{l}-x_{\ob}|)(s-1)}}{}\dv x_{l}\dv x_{\ob},
\end{align*}where
step $(a)$ is obtained due to the symmetry of the function $e^{m\bar{\A}(t/2,a,|x_{l}-x_{\ob}|)(s-1)}$. The term $|x_{\road}-x_{\ob}|$ exhibits symmetry with respect to one variable, \ie, $f(-x,y)=f(x,y)$. To further simplify the integral let $y=|z|,\, z=x_{\road}-x_{\ob}$. The distribution of $z$ is the convolution of $x_{\road}$ with $-x_{\ob}$. Since the distribution of $x_{\ob}$ and $-x_{\ob}$ is same, the PDF of $z$ is the convolution of the PDF of $x_{\road}$ with $x_{\ob}$. Hence, the PDF of $z$ is
	\begin{align*}
		&f_{z}(x)=f_{x_{\road}}(x)\ast f_{x_{\ob}}(x)\\
		&=\left(\left(u(x)-u(x-t/2)\right)\ast\left(u(x+a)-u(x-a)\right)\right)/at\\
		&=\left({\mathcal{R}(x+a)-\mathcal{R}(x-a)}-\mathcal{R}(x+a-t/2)\right.\\
		&\left.+\mathcal{R}(x-t/2-a)\right)/(at),
	\end{align*}where $\mathcal{R}(\cdot)$ is the ramp function defined in \eqref{rampsignal}. From above, we may easily determine the PDF of $y$ given as
\begin{align*}
	f_{y}(x)=f_{z}(x)+f_{z}(-x),
\end{align*}which can be further simplified with final expression $f_{y}(x)$ is provided in \eqref{distr_fy}. 
Using the distribution of $y$ the PGF $\mathcal{P}_{V_{\mrm}(t/2)}(s)$ for $a>t/2$
	\begin{align*}
		&I=\frac{1}{at}\int_{0}^{a-t/2}e^{m\bar{\beta}(t/2,a)(s-1)}t\dv y\\
		&+\frac{1}{at}\int_{a-t/2}^{a+t/2}e^{(m/(2a))(s-1)(a+t/2-y)}(t-y+a-t/2)\dv y\\
		&\stackrel{(a)}=\frac{e^{m\bar{\beta}(t/2,a)(s-1)}t(a-t/2)}{at}\\
		&+\frac{1}{at}\int_{a-t/2}^{a+t/2}\!\!\!\! e^{(m/(2a))(s-1)(a+t/2-y)}
		(t/2+a-y)\dv y\\
		&\stackrel{(b)}=\frac{(a-t/2)e^{m\bar{\beta}(t/2,a)(s-1)}}{a}+\frac{1}{at((m/(2a))(s-1))^{2}}\times\\
		&\quad\int_{0}^{(m/(2a))(s-1)t} e^{z}{z}\dv z\\
		&\stackrel{(c)}=\frac{(a-t/2)e^{m\bar{\beta}(t/2,a)(s-1)}}{a}+\frac{1}{at((m/(2a))(s-1))^{2}}\times\\
		&\quad\left(e^{(m/(2a))(s-1)t}((m/(2a))(s-1)t-1)+1\right),
	\end{align*}where step $(a)$ is obtained by simplifying the first integral, step $(b)$ is obtained by replacing $(t/2+a-y)=z\implies -\dv y= \dv z$, finally step $(c)$ is obtained by solving the integral after the substitution. Similarly, when  $t/2>a,$ the integral simplifies to
	\begin{align*}
		I_{1}&=\frac{1}{at}\int_{0}^{t/2-a}e^{m\bar{\beta}(t/2,a)(s-1)}2a\dv y+\frac{1}{at}\times\\
		&\int_{t/2-a}^{t/2+a}e^{(m/(2a))(s-1)(a+t/2-y)}(2a-(y-t/2+a))\dv y\\
		&\stackrel{(a)}=\frac{2}{t}e^{m\bar{\beta}(t/2,a)(s-1)}(t/2-a)+\frac{1}{at(m/(2a))^{2}(s-1)^{2}}\times\\
		&\left(e^{m(s-1)}(m(s-1)-1)+1\right),
	\end{align*}where step $(a)$ is obtained by simplifying the integral using the similar technique for the case $a>t/2$.
	\section{}\label{proofofthirdmoment}
	From Theorem \ref{thm:3}, the PGF of $V_{\mrm}(t/2)$ for $a>t/2$ is
	\begin{align*}
		&\mathcal{P}_{V_{\mrm}(t/2)}(s)=\frac{(a-t/2)e^{\lambda_{\drm}t(s-1)}}{a}+\frac{\lambda_{\drm}t(s-1)e^{\lambda_{\drm}t(s-1)}}{at\lambda_{\drm}^{2}(s-1)^{2}}\\
		&+\frac{1-e^{\lambda_{\drm}t(s-1)}}{at\lambda_{\drm}^{2}(s-1)^{2}}\\
		&\stackrel{(a)}=\frac{(a-t/2)e^{\lambda_{\drm}t(s-1)}}{a}+\frac{\lambda_{\drm}t(s-1)e^{\lambda_{\drm}t(s-1)}}{at\lambda_{\drm}^{2}(s-1)^{2}}\\
		&-\frac{\lambda_{\drm}t(s-1)+\frac{(\lambda_{\drm}t(s-1))^{2}}{2!}+\frac{(\lambda_{\drm}t(s-1))^{3}}{3!}+\cdots}{at\lambda_{\drm}^{2}(s-1)^{2}}\\
		&\stackrel{(b)}=\frac{(a-t/2)e^{\lambda_{\drm}t(s-1)}}{a}+\frac{e^{\lambda_{\drm}t(s-1)}-1-\frac{\lambda_{\drm}t(s-1)}{2!}-\cdots}{a\lambda_{\drm}(s-1)}\\
		&\stackrel{(c)}={\left(1-{t}/{(2a)}\right)e^{\lambda_{\drm}t(s-1)}}\\
		&+\frac{\left(t-\frac{t}{2!}+\left(\frac{t^{2}\lambda_{\drm}(s-1)}{2!}-\frac{t^{2}(\lambda_{\drm}(s-1))}{3!}\right)\right)\cdots}{a},
	\end{align*}where step $(a)$ is obtained using the Taylor series expansion of exponential function, step $(b)$ and $(c)$ are the simple algebraic manipulations to find the limit $s\rightarrow1$. Now, to derive the mean and the variance of $\mathcal{P}_{V_{\mrm}(t/2)}(s)$, we need the first and second derivative of the $\mathcal{P}_{V_{\mrm}(t/2)}(s)$ which is 
	\begin{align*}
		&\lim_{s\rightarrow1}\mathcal{P}^{(1)}_{V_{\mrm}(t/2)}(s)\!=\!(1-{1}/{(2a)})\lambda_{\drm}t+({t^{2}\lambda_{\drm}}/{2}-{t^2\lambda_{\drm}}/{3!})/{a}.
	\end{align*}Similarly, we can find $\lim_{s\rightarrow1}\mathcal{P}^{(2)}_{V_{\mrm}(t/2)}(s)$ which is 
	\begin{align*}
		&=\left(1-{t}/{(2a)}\right)\left(\lambda_{\drm}t\right)^{2}+({1}/{a})\left({t^{3}\lambda_{\drm}^{2}}/{3}-{t^{3}\lambda_{\drm}^{2}}/{12}\right).
	\end{align*}Using \eqref{variance}, we get the variance $\mathrm{Var}[V_{\mrm}(t/2)]$ as
	\begin{align*}
		&=\left(1-{t}/{(4a)}\right)\left(\lambda_{\drm}t\right)^{2}+\left(\left(1-{t}/{(2a)}\right)\lambda_{\drm}t+{\lambda_{\drm} t^{2}}/{(6a)}\right)\\
		&-\left(\left(1-{t}/{(2a)}\right)\lambda_{\drm}t+{\lambda_{\drm} t^{2}}/{(6a)}\right)^{2}.
	\end{align*}
	Similarly, for $t/2>a$, we can derive the  mean, variance and third moment of $V_{\mrm}(t/2)$.
	{\section{}\label{proof_of_theorem3}
		The PGF of $\Phi_{\mrm}\left(\bt_{1}(\ob,L_{\ob}/2)\right)$ can be obtained from \eqref{PGFMCP}. For the PGF of $\Phi_{\bm{x}_{0}}(\bt_{1}(\ob,a)\cap L_{\ob})$, the parent $\bm{x}_{0}$ of the typical point is uniformly located in ball $\bt_{1}(\ob,a)$. Let the center of the tagged cell be denoted by $x_{l}$ which is uniformly distributed between $[-L_{\ob}/2,\,L_{\ob}/2]$. Conditioned on the location of $x_{l}$ and $\bm{x}_{0}$, the number of points of tagged platoon falling in tagged cell is Poisson distributed with mean $(m/(2a))\A_{1}(L_{\ob}/2,a,|x_{l}-\bm{x}_{0}|)$ where $\A_{1}(\cdot,\cdot,\cdot)$ denotes the length of the segment of intersection between tagged platoon with the tagged cell. Deconditioning with distribution of $\bm{x}_{0}$ and $x_{l}$ and taking product with the PGF of $S(t/2)$ and finally deconditioning using the PDF of $L_{\ob}$, we get the PGF of $\widetilde{S}_{\pt}$. Further using the PGF, we can easily obtain the PMF of $\widetilde{S}_{\pt}$. }

\section{}\label{proof_of_nodedegree}
The total number of VUs within the communication range $R_{\brm}$ of the typical VU located at the origin is the sum of the following two RVs. First is the number $S$ of VUs within the communication range of the typical VU, excluding the VUs of the tagged platoon associated with the typical VU. The PGF for this RV is
	$\mathcal{P}_{S}(s,R_{\brm}/2)$.
	The second RV is the number of VUs of the tagged platoon falling within communication range of the typical VU, for which the PGF is $\frac{1}{a}\int_{0}^{a}e^{\lambda_{\drm}\A(R_{\brm}/2,a,x)(s-1)}\dv x$. Since these two are independent RVs, the PGF of connectivity degree is the product of individual PGFs. For N-PTS, applying Slivinayak's theorem and excluding the typical VU, the remaining VUs within its communication range are Poisson distributed with mean $\lambda R_{\brm}$, enabling the straightforward derivation of the PGF for N-PTS.
\section{}\label{proof-meta-cov}
 Replacing $q=\irm t$ in MGF presented in \eqref{MGF1}, we get $M_{\irm t}(\tau)$ 
	\begin{align*}
&=\int_{r=0}^{\infty}\expS{-\frac{2\prm_{\cdot}\lambda_{\rsu} r}{\alpha}\int_{0}^{1}\left(1-{\left(1+\tau y\right)^{-\irm t}}\right)y^{-\eta}\dv y}\\
&\quad\times e^{-\frac{\irm t\tau r^{\alpha}}{\SNR}}f_{R}(r)\dv r\\
		&\stackrel{(a)}=\int_{r=0}^{\infty}\expS{-\frac{2\prm_{\cdot}\lambda_{\rsu} r}{\alpha}\int_{0}^{1}\left(1-e^{-\irm t \ln\left(1+\tau y\right)}\right)y^{-\eta}\dv y}\\
		&\quad e^{-\frac{\irm t\tau r^{\alpha}}{\SNR}}f_{R}(r)\dv r\\
		&\stackrel{(b)}=\int_{r=0}^{\infty}\exp\left(-\frac{2\prm_{\cdot}\lambda_{\rsu} r}{\alpha}\int_{0}^{1}\left(1-\cos(t\ln (1+\tau y))\right.\right.\\
		&\left.\left.+\irm \sin\left(t\ln\left(1+\tau y\right)\right)\right)y^{-\eta}\dv y\right)e^{-\frac{\irm t\tau r^{\alpha}}{\SNR}}f_{R}(r)\dv r\\
		&\stackrel{(c)}=\int_{r=0}^{\infty}\expS{-\frac{2\prm_{\cdot}\lambda_{\rsu} r}{\alpha}\int_{0}^{1}\left(1-\cos\left(t\ln\left(1+\tau y\right)\right)\right)y^{-\eta}\dv y}\\
		&\expS{\!\!-\irm \frac{2\prm_{\cdot}\lambda_{\rsu} r}{\alpha}\!\!\int_{0}^{1}\!\!\!\!\sin\left(t\ln\left(1+\tau y\right)\right)y^{-\eta}\dv y} e^{-\frac{\irm t \tau r^{\alpha}}{\SNR}}f_{R}(r)\dv r,
	\end{align*}where step $(a)$ follows by substituting  $(1+\tau y)^{\irm t}$ with $e^{-\irm t \ln\left(1+\tau y\right)}$, step $(b)$ is obtained using the Euler's formula and step $(c)$ rewrites the result from step $(b)$ as a product of real and the complex terms. Simplifying further $\mathrm{Im}\left(e^{-\irm t \ln (x)}M_{\irm t}(\tau)\right)$ is
	\begin{align}\label{replacment}
		&=-\int_{r=0}^{\infty}f_{\mathrm{r}}(t)\sin\left(f_{\irm}(t)+\Theta(t)\right)f_{R}(r)\dv r,
	\end{align}where $f_{\mathrm{r}}(t)$, $f_{\irm}(t)$ and $\Theta(t)$ is defined in \eqref{fr}, \eqref{fi} and \eqref{the}, respectively. Finally, replacing $\mathrm{Im}\left(e^{-\irm t \ln (x)}M_{\irm t}(\tau)\right)$ in \eqref{replacment} to \eqref{MD proof} and simplifying further completes the proof of the theorem. 
		\vspace{-1em}
	
\bibliographystyle{ieeetran}
\bibliography{scibib,NewCompPaperDB}

\begin{thebibliography}{10}
\providecommand{\url}[1]{#1}
\csname url@samestyle\endcsname
\providecommand{\newblock}{\relax}
\providecommand{\bibinfo}[2]{#2}
\providecommand{\BIBentrySTDinterwordspacing}{\spaceskip=0pt\relax}
\providecommand{\BIBentryALTinterwordstretchfactor}{4}
\providecommand{\BIBentryALTinterwordspacing}{\spaceskip=\fontdimen2\font plus
\BIBentryALTinterwordstretchfactor\fontdimen3\font minus
  \fontdimen4\font\relax}
\providecommand{\BIBforeignlanguage}[2]{{%
\expandafter\ifx\csname l@#1\endcsname\relax
\typeout{** WARNING: IEEEtran.bst: No hyphenation pattern has been}%
\typeout{** loaded for the language `#1'. Using the pattern for}%
\typeout{** the default language instead.}%
\else
\language=\csname l@#1\endcsname
\fi
#2}}
\providecommand{\BIBdecl}{\relax}
\BIBdecl

\bibitem{jia2015survey}
D.~Jia, K.~Lu, J.~Wang, X.~Zhang, and X.~Shen, ``A survey on platoon-based
  vehicular cyber-physical systems,'' \emph{IEEE Commun. Surveys Tuts.},
  vol.~18, no.~1, pp. 263--284, 2015.

\bibitem{hussein2021vehicle}
A.~A. Hussein and H.~A. Rakha, ``Vehicle platooning impact on drag coefficients
  and energy/fuel saving implications,'' \emph{IEEE Trans. Veh. Technol.},
  vol.~71, no.~2, pp. 1199--1208, 2021.

\bibitem{perfecto2017millimeter}
C.~Perfecto, J.~Del~Ser, and M.~Bennis, ``Millimeter-wave {V2V} communications:
  Distributed association and beam alignment,'' \emph{IEEE J. Sel. Areas
  Commun.}, vol.~35, no.~9, pp. 2148--2162, 2017.

\bibitem{pandey2023properties}
K.~Pandey, A.~K. Gupta, H.~S. Dhillon, and K.~R. Perumalla, ``Properties of a
  random bipartite geometric associator graph inspired by vehicular networks,''
  \emph{Entropy}, vol.~25, no.~12, p. 1619, 2023.

\bibitem{pandey2023fundamentals}
K.~Pandey, K.~R. Perumalla, A.~K. Gupta, and H.~S. Dhillon, ``Fundamentals of
  vehicular communication networks with vehicle platoons,'' \emph{IEEE Trans.
  Wireless Commun.}, vol.~22, no.~12, pp. 8634--8649, 2023.

\bibitem{huang2020geometry}
C.~Huang, R.~Wang, P.~Tang, R.~He, B.~Ai, Z.~Zhong, C.~Oestges, and A.~F.
  Molisch, ``Geometry-cluster-based stochastic mimo model for
  vehicle-to-vehicle communications in street canyon scenarios,'' \emph{IEEE
  Trans. Wireless Commun.}, vol.~20, no.~2, pp. 755--770, 2020.

\bibitem{sun2020distributed}
Y.~Sun, S.~Zhou, and Z.~Niu, ``Distributed task replication for vehicular edge
  computing: Performance analysis and learning-based algorithm,'' \emph{IEEE
  Trans. Wireless Commun.}, vol.~20, no.~2, pp. 1138--1151, 2020.

\bibitem{kimura2021performance}
T.~Kimura, ``Performance analysis of cellular-relay vehicle-to-vehicle
  communications,'' \emph{IEEE Trans. Veh. Technol.}, vol.~70, no.~4, pp.
  3396--3411, 2021.

\bibitem{dhillon2020poisson}
H.~S. Dhillon and V.~V. Chetlur, \emph{{P}oisson Line {C}ox Process:
  Foundations and Applications to Vehicular Networks}.\hskip 1em plus 0.5em
  minus 0.4em\relax Morgan \& Claypool Publishers, 2020.

\bibitem{chetlur2020load}
V.~V. Chetlur and H.~S. Dhillon, ``On the load distribution of vehicular users
  modeled by a {P}oisson line {C}ox process,'' \emph{IEEE Wireless Commun.
  Lett.}, vol.~9, no.~12, pp. 2121--2125, 2020.

\bibitem{choi2018poisson}
C.-S. Choi and F.~Baccelli, ``Poisson {C}ox point processes for vehicular
  networks,'' \emph{IEEE Trans. Veh. Technol}, vol.~67, no.~10, pp.
  10\,160--10\,165, 2018.

\bibitem{blaszczyszyn2013stochastic}
B.~B{\l}aszczyszyn, P.~M{\"u}hlethaler, and Y.~Toor, ``Stochastic analysis of
  {ALOHA} in vehicular ad hoc networks,'' \emph{Annals of
  telecommunications-Annales des t{\'e}l{\'e}communications}, vol.~68, no.~1,
  pp. 95--106, 2013.

\bibitem{1dchapter}
S.~Kwon, Y.~Kim, and N.~B. Shroff, ``Analysis of connectivity and capacity in
  1-{D} vehicle-to-vehicle networks,'' \emph{IEEE Trans. Wireless Commun.},
  vol.~15, no.~12, pp. 8182--8194, 2016.

\bibitem{cheng2020connectivity}
J.~Cheng, G.~Yuan, M.~Zhou, S.~Gao, Z.~Huang, and C.~Liu, ``A
  connectivity-prediction-based dynamic clustering model for {V}{A}{N}{E}{T} in
  an urban scene,'' \emph{IEEE Internet Things J.}, vol.~7, no.~9, pp.
  8410--8418, 2020.

\bibitem{tong2016stochastic}
Z.~Tong, H.~Lu, M.~Haenggi, and C.~Poellabauer, ``A stochastic geometry
  approach to the modeling of {D}{S}{R}{C} for vehicular safety
  communication,'' \emph{IEEE Trans. Intell. Transp. Syst.}, vol.~17, no.~5,
  pp. 1448--1458, 2016.

\bibitem{koufos2019meta}
K.~Koufos and C.~P. Dettmann, ``The meta distribution of the {SIR} in linear
  motorway {VANETs},'' \emph{IEEE Trans. Commun.}, vol.~67, no.~12, pp.
  8696--8706, 2019.

\bibitem{ni2015interference}
M.~Ni, J.~Pan, L.~Cai, J.~Yu, H.~Wu, and Z.~Zhong, ``Interference-based
  capacity analysis for vehicular ad hoc networks,'' \emph{IEEE Commun. Lett.},
  vol.~19, no.~4, pp. 621--624, 2015.

\bibitem{shao2015performance}
C.~Shao, S.~Leng, Y.~Zhang, A.~Vinel, and M.~Jonsson, ``Performance analysis of
  connectivity probability and connectivity-aware {MAC} protocol design for
  platoon-based {VANETs},'' \emph{IEEE Trans. Veh. Technol}, vol.~64, no.~12,
  pp. 5596--09, 2015.

\bibitem{wang2022design}
Z.~Wang, S.~Jin, L.~Liu, C.~Fang, M.~Li, and S.~Guo, ``Design of intelligent
  connected cruise control with vehicle-to-vehicle communication delays,''
  \emph{IEEE Trans. Veh. Technol.}, vol.~71, no.~8, pp. 9011--9025, 2022.

\bibitem{1dpaper}
S.~Kwon, Y.~Kim, and N.~B. Shroff, ``Analysis of connectivity and capacity in
  1-d vehicle-to-vehicle networks,'' \emph{IEEE Trans. Wireless Commun.},
  vol.~15, no.~12, pp. 8182--8194, 2016.

\bibitem{SGBook2022}
J.~G. Andrews, A.~K. Gupta, A.~Alammouri, and H.~S. Dhillon, \emph{An
  {I}ntroduction to {C}ellular {N}etwork {A}nalysis using {S}tochastic
  {G}eometry}.\hskip 1em plus 0.5em minus 0.4em\relax Morgan Claypool
  (Springer), 2023.

\bibitem{tanemura2003statistical}
M.~Tanemura, ``Statistical distributions of {P}oisson {V}oronoi cells in two
  and three dimensions,'' \emph{Forma}, vol.~18, no.~4, pp. 221--247, 2003.

\bibitem{tagged1d}
A.~Gupta, ``On the distribution of various {V}oronoi association cells in
  cellular networks,'' 2021, available at
  \url{https://home.iitk.ac.in/~gkrabhi/}.

\bibitem{pandeykth}
K.~{Pandey} and A.~K. {Gupta}, ``$k$th {D}istance distributions of
  $n$-dimensional {M}at{\'e}rn cluster process,'' \emph{IEEE Commun. Lett.},
  vol.~25, no.~3, pp. 769--773, 2021.

\bibitem{faadibruno}
W.~P. Johnson, ``The curious history of {F}a{\`a} di {B}runo's formula,''
  \emph{The American Mathematical Monthly}, vol. 109, no.~3, pp. 217--234,
  2002.

\bibitem{gupta2015sinr}
A.~K. Gupta, X.~Zhang, and J.~G. Andrews, ``{S}{I}{N}{R} and throughput scaling
  in ultradense urban cellular networks,'' \emph{IEEE Wireless Commun. Lett.},
  vol.~4, no.~6, pp. 605--608, 2015.

\bibitem{haenggi2015meta}
M.~Haenggi, ``The meta distribution of the {SIR} in {P}oisson bipolar and
  cellular networks,'' \emph{IEEE Trans. Wireless Commun.}, vol.~15, no.~4, pp.
  2577--2589, 2015.

\bibitem{Gil1951note}
J.~Gil-Pelaez, ``Note on the inversion theorem,'' \emph{Biometrika}, vol.~38,
  no. 3-4, pp. 481--482, 1951.

\bibitem{tensee}
N.~H. T.~S. Administration, ``Summary of state speed laws ninth edition:
  Current as of october 8, 2012,'' \emph{National Committee on Uniform Traffic
  Laws and Ordinances: Washington, DC, USA}, 2012.

\bibitem{shi2020effects}
J.~Shi, C.~Wu, and X.~Qian, ``The effects of multiple factors on elderly
  pedestrians’ speed perception and stopping distance estimation of
  approaching vehicles,'' \emph{Sustainability}, vol.~12, no.~13, p. 5308,
  2020.

\end{thebibliography}
\vspace{12pt}
\color{red}
\end{document}